\documentclass[acmsmall, nonacm]{acmart}

\AtBeginDocument{%
  }

\setcopyright{acmlicensed}
\copyrightyear{2018}
\acmYear{2018}
\acmDOI{XXXXXXX.XXXXXXX}

\acmConference[Conference acronym 'XX]{Make sure to enter the correct
  conference title from your rights confirmation emai}{June 03--05,
  2018}{Woodstock, NY}
\acmISBN{978-1-4503-XXXX-X/18/06}

\usepackage[ruled, linesnumbered]{algorithm2e}

\usepackage{array} 

\usepackage{graphicx}
\usepackage{listings}

\usepackage{siunitx} 

\newcommand{\code}[1]{\lstinline!#1!}

\usepackage{tikz}
\usetikzlibrary{automata, positioning}

\usepackage{listings}

\theoremstyle{definition}

\newtheorem{theorem}{Theorem}[section]
\newtheorem{definition}{Definition}[section]

\newtheorem{lemma}{Lemma}[section]

\newtheorem{proposition}{Proposition}[section]

\newcommand{\callsym}[1]{\text{\guilsinglleft{\(#1\)}}}
\newcommand{\retsym}[1]{ \text{{\(#1\)}\guilsinglright}}

\newcommand{\oracle}{\mathcal{O}}
\newcommand{\true}{\text{True}}
\newcommand{\false}{\text{False}}

\usepackage{enumitem}

\newcommand{\ac}{\callsym {a}}
\newcommand{\bc}{\retsym{b}}

\newcommand{\ta}{{h_a}}
\newcommand{\tb}{{h_b}}

\usepackage{graphicx}

\newcommand{\stack}{T}

\newcommand{\ab}[1]{{\ac {#1} \bc}}

\newcommand{\plainsyms}{\Sigma_{\text{plain}}}
\newcommand{\callsyms}{\Sigma_{\text{call}}}
\newcommand{\retsyms}{\Sigma_{\text{ret}}}

\newcommand{\lang}{\mathcal{L}}

\newcommand{\automaton}{\mathcal{H}}
\newcommand{\vpa}{\mathcal{H}}
\newcommand{\fsa}{\mathcal{H}}

\newcommand{\glade}{\textsc{Glade}\xspace}
\newcommand{\arvada}{\textsc{Arvada}\xspace}
\newcommand{\lstar}{{\( L^* \)}\xspace}
\newcommand{\vstar}{{V-Star}\xspace}
\newcommand{\ttt}{{TTT}\xspace}
\newcommand{\reinam}{{REINAM}\xspace}

\newcommand{\ksevpa}{\( k \)-SEVPA\xspace}

\newcommand{\eqlstar}{\simeq}

\newcommand{\qc}{\mathcal{Q}}

\newcommand{\memquery}{\chi}
\newcommand{\eqquery}{\mathcal{E}}

\newcommand{\tk}{{\tau}}
\newcommand{\convert}{{\mathbf{conv}}}

\newcommand{\arta}{\triangleleft}
\newcommand{\artb}{\triangleright}

\newcommand{\edited}[1]{{\color{blue}{#1}}}
\newcommand{\secondEdited}[1]{{\color{blue}{#1}}}
\renewcommand{\edited}[1]{{#1}}
\renewcommand{\secondEdited}[1]{{#1}}

\usepackage{math-cmds}

\begin{document}

\title{\vstar: Learning Visibly Pushdown Grammars from Program Inputs (Extended Version)}


\author{Xiaodong Jia}
\email{xxj34@psu.edu}
\orcid{0000-0003-2493-9111}
\author{Gang Tan}
\email{gtan@psu.edu}
\orcid{0000-0001-6109-6091}
\affiliation{%
  \institution{The Pennsylvania State University}
  \streetaddress{201 Old Main}
  \city{State College}
  \state{Pennsylvania}
  \country{USA}
  \postcode{16802}
}

\renewcommand{\shortauthors}{Jia and Tan, et al.}

\begin{abstract}
    Accurate description of program inputs remains a critical challenge in the field of programming languages. Active learning, as a well-established field, achieves exact learning for regular languages. We offer an innovative grammar inference tool, V-Star, based on the active learning of visibly pushdown automata. V-Star deduces nesting structures of program input languages from sample inputs, employing a novel inference mechanism based on nested patterns. This mechanism identifies token boundaries and converts languages such as XML documents into VPLs. We then adapted Angluin's L-Star, an exact learning algorithm, for VPA learning, which improves the precision of our tool. Our evaluation demonstrates that V-Star effectively and efficiently learns a variety of practical grammars, including S-Expressions, JSON, and XML, and outperforms other state-of-the-art tools.
\end{abstract}



\maketitle

\section{Introduction}

In recent years, there has been a growing interest in learning
grammars from a set of sample strings. This
interest stems from a wide range of applications in fuzzing, program
validation, and other areas~\cite{Arefin24,10.1145/3140587.3062349,10.1109/ASE51524.2021.9678879,10.1145/3338906.3338958,10.1145/3519939.3523716}. 
Despite significant progress,
the challenge of learning the input grammars for black-box programs
remains, particularly when considering grammars with inherent
complexities. This challenge is part of a broader problem that has
been extensively studied for regular languages but is significantly
more difficult when dealing with broader classes of grammars.

Recently, GLADE~\cite{10.1145/3140587.3062349} (followed by a replication
study~\cite{10.1145/3519939.3523716}) and ARVADA~\cite{10.1109/ASE51524.2021.9678879} have been proposed to learn context-free grammars
(CFGs) under active learning. Both approaches require positive seed
inputs and utilize enumeration and heuristics to reduce the search
space.
However, methods in CFG learning
such as \glade and \arvada do not fully utilize the concept of \emph{nesting
structures}, which could potentially improve a grammar-learning process's accuracy.

Nesting structures are widely observed in practical languages, where recursions are 
explicitly delimited in their sentences. 
For example, an XML document's open and close tags
delimit a component of the document and can be nested within other
open and close tags. These nesting structures often carry valuable
insights into the grammars' structure and could potentially be a
powerful tool for learning grammars. 


To achieve accurate learning, we model the 
nesting structures together with the target language as
Visibly Pushdown Grammars (VPGs)~\cite{10.1145/1516512.1516518}, a subclass of
CFGs. VPGs formally specify nesting structures, and despite being
slightly weaker than CFGs they can specify many practical format
languages such as XML and JSON; they also enjoy all desirable closure
properties; e.g., the set of visibly pushdown languages is closed
under intersection, concatenation, and complement~\cite{10.1145/1516512.1516518}.
We posit that these properties position VPGs as an ideal mechanism for
learning practical grammars, which is the focus of this paper.

Learning VPGs is a problem that fits nicely into the well-studied 
\emph{active learning} field. In this field,
Angluin~\cite{10.1016/0890-5401(87)90052-6} first demonstrated that it is possible to
efficiently learn regular languages from a minimally adequate
teacher (MAT), which answers (1) whether a string is in the language,
which is called a membership query, and (2) whether a finite state
automaton accepts exactly the language held by the teacher, which
is called an equivalence query; if not, the teacher provides a
counterexample accepted by either the automaton or the language, but
not both. 

\edited{Despite advancements in VPG learning facilitated by a MAT, current techniques, as highlighted in previous work~\cite{pmlr-v153-barbot21a,michaliszyn_et_al:LIPIcs.MFCS.2022.74,Isberner2015FoundationsOA}, are constrained by assumptions that
do not match practical settings. Specifically, they
assume the availability of known nesting patterns, or more technically, a predefined \emph{tagging function}---a cornerstone of the VPG formalism detailed in Section~\ref{sec:Background}. This tagging function defines a set of \emph{call} and \emph{return} symbols and is assumed to operate on individual characters. In contrast, in many practical settings the tagging must be inferred and it operates on sequences of characters (i.e., tokens) rather than on individual characters; a more detailed comparison with prior work and the limitations of these assumptions are discussed in Section~\ref{sec:RelatedWork}.}


\edited{To address these limitations,
  we introduce \vstar, a novel grammar-inference framework designed to
  learn VPGs from a black-box program using a collection of sample
  seed strings. \vstar's algorithm is inspired by
  \lstar~\cite{10.1016/0890-5401(87)90052-6}, which learns a finite-state automaton using a MAT
  and seed strings. We develop \vstar in several pivotal
  steps. First, we develop an \lstar{}-like algorithm for learning
  VPGs when tags are known and are on single characters; this
  algorithm utilizes \ksevpa~\cite{Alur2005Congruences} to define
  a set of congruence relations, critical to the algorithm. Second, we
  develop a tag-inference algorithm, which utilizes a novel notion of \emph{nesting patterns}
  to infer call and return symbols, assuming they are of single
  characters.  In the third step, we lift the restriction that tags
  are on single characters and develop an algorithm for inferring
  token-based tags. Finally, we remove the requirement of equivalence
  queries by simulating them using membership queries via sampling.
  These steps all together result in a practical framework for learning VPGs
  from seed strings.  The main contributions of \vstar are summarized as follows:
}


\begin{enumerate}
    \item \textbf{Innovative Tool and Methodology:}
      \vstar is a novel tool for VPG inference. Its algorithm adapts
      Angluin's \lstar algorithm and integrates
      a set of novel techniques such as nesting patterns to \edited{infer}
      call and return tokens.  To our best knowledge, this is the
      first VPG-learning algorithm without knowing what call/return
      tokens are a priori.

    \item \textbf{Theoretical Reasoning:} We provide a theoretical
      analysis elucidating the conditions under which \vstar
      achieves accurate learning. 
      \edited{We first prove that for any character-based visibly pushdown language with uniquely paired call/return symbols, there exists a finite set of seed strings from which 
      \vstar can learn a tagging function that achieves exact learning.
      We further show that under some realistic assumptions \vstar can infer tagged tokens for token-based visibly pushdown languages.}

    \item \textbf{{Accuracy}:} Our evaluation of \vstar demonstrates
      its better accuracy in learning practical grammars,
      in comparison with state-of-the-art grammar learning tools. The
      accuracy result highlights the benefit of utilizing the concept
      of nesting structures \edited{and \vstar's ability to simulate the equivalence queries by sampling test strings from the seed strings.}
\end{enumerate}

\section{Related Work}\label{sec:RelatedWork}
In the realm of automata learning, learning finite state automata is a
well-studied field. The \lstar Algorithm by Angluin~\cite{10.1016/0890-5401(87)90052-6},
which learns a finite state automaton held by a MAT in polynomial time, is a seminal work in this
area. Following \lstar, various adaptations focusing on the active
learning problem have been proposed, such as \cite{rivest1989inference,10.5555/114836.114848, irfan2010angluin, howar2012active, Isberner2015FoundationsOA}, to list a few.

\glade~\cite{10.1145/3140587.3062349}, targeting learning a CFG from an oracle,
employs a two-step algorithmic approach. Initially, it enumerates
all substrings of seed strings and attempts to generalize
these substrings into regular expressions. Subsequently, nonterminals
are created and merged based on learned regular expressions.
\arvada~\cite{10.1109/ASE51524.2021.9678879} also aims to learn a CFG by using a
technique that exchanges two substrings from seed
strings. If two substrings are interchangeable, they are assigned
the same nonterminal, and the process gradually constructs parse
trees. \arvada employs heuristics to exchange substrings with similar
contexts. In its evaluation, the context is considered to be the
surrounding substrings of length four. 
While context strings bear some similarity to call and return
symbols in VPGs, 
call and return symbols are more flexible than context strings in \arvada: 
they can wrap contexts of an arbitrary length.
Moreover, call and return symbols have a stronger implication
on the recursive structure of the oracle language, a feature that
\vstar capitalizes on for better grammar-inference accuracy, which is evidenced by
\vstar's experimental comparison with \glade and \arvada discussed
in Section~\ref{sec:Evaluation}.
As an extension of GLADE, \reinam~\cite{10.1145/3338906.3338958} refines the grammar
learned by GLADE using reinforcement learning. This process allows for
the potential replacement of GLADE by other learning tools such as
\arvada~\cite{10.1109/ASE51524.2021.9678879} or our tool \vstar.

Learning VPGs, a subset of CFGs, has been the
focus of several systems. For example, assuming that the set of call and
return symbols is known, VPL*~\cite{pmlr-v153-barbot21a} learns VPGs
with membership and equivalence queries. The approach taken by VPL* is
indirect, first using TL* \cite{Drewes2007-tcs,drewes_mat_2011} to learn
a tree automaton, which is then converted to a VPG. Another work by
\citet{michaliszyn_et_al:LIPIcs.MFCS.2022.74} also assumes that the set of call and return
symbols is known and attempts to learn a visibly pushdown automaton
(VPA). Moreover, it requires a stronger teacher who, in addition to
providing membership and equivalence queries, can also report the
stack content during the VPA's
execution. TTT~\cite{Isberner2015FoundationsOA} is another VPA-inference tool
under the active-learning setting, based on
discrimination trees. \vstar differentiates itself from these systems
by learning call and return symbols from the oracle and seed strings.

In practice, a MAT is often instantiated
by a black-box program and the oracle language comprises input strings
that do not trigger program errors (assuming the program always
terminates). Thus, membership queries require only program
execution. However, equivalence queries are much harder to answer.
Simulating equivalence queries has a long history, often under the
name of \emph{conformance testing}
\cite{10.1145/3605360}. Chow \cite{10.1109/TSE.1978.231496} first
proposed the so-called \emph{W-method} for Mealy machines; we note
that the FSA version of the W-method is essentially a brute force
approach of enumerating suffix strings that distinguish the
representative prefix strings in the Nerode relation. The space of
suffix strings is restricted from \( \Sigma^* \) to \( \Sigma^k \), under the
assumption that the difference between the size of the oracle FSA and
the size of the learned FSA is \( k \). The W-method has many variants
\cite{vasilevskii_failure_1973,10.1145/1081180.1081189,87284}. 
For more information, we
refer to \citet{10.1145/3605360}.






\section{Background}\label{sec:Background}

\subsection{Grammar Inference}\label{sec:ProblemStatement}
In grammar inference, we assume an \emph{oracle} for the grammar being inferred, denoted as \(\oracle\), which maps strings to booleans---true for valid strings and false for invalid ones. The set of valid strings according to the oracle forms the \emph{oracle language}, denoted as \(\lang_{\mathcal{O}}\). When the oracle language is defined by a grammar, we refer to this grammar as the \emph{oracle grammar}, denoted as \(G_{\mathcal{O}}\). Similarly, when the oracle language can be recognized by a deterministic finite automaton (DFA), we refer to this DFA as the \emph{oracle automaton}, denoted as \(\automaton_{\mathcal{O}}\).

We define the active learning problem as follows:

\textit{Inputs:} The problem takes two inputs: 
\begin{enumerate}
    \item A set \(\Sigma\) of terminals, and
    \item A minimally adequate teacher (MAT), which can answer both
      membership and equivalence queries.  For an equivalence query
      with a hypothesis grammar, the teacher returns true when the
      hypothesis grammar is equivalent to the oracle grammar, meaning
      they generate the same language; if they are not equivalent, it
      provides a counterexample, which is a string accepted by either
      the hypothesis grammar or the oracle grammar, but not by both.
\end{enumerate}


\textit{Output:}
The goal of the active learning problem is to construct a grammar, denoted as \(\mathcal{G}\), such that the language it generates, denoted as \(\lang_{\mathcal{G}}\), is identical to the oracle language \(\lang_{\mathcal{O}}\). 

\subsection{Visibly Pushdown Grammars}
The expressive power of VPGs~\cite{10.1145/1516512.1516518} is between regular
grammars and context-free grammars, and VPGs are sufficient for
describing the syntax of many practical languages, such as JSON, XML,
and HTML.
Application wise, VPGs have been used in program analysis, XML
processing, and other fields \cite{10.1007/978-3-642-31424-7_41,heizmann2010nested,10.5555/1770532.1770559,nguyen2006vpa,kumar2007visibly,alur2007marrying,thomo2008rewriting,GAUWIN200813,10.14778/1920841.1920865,mozafari2012high}. {Besides, since they can be efficiently parsed, VPGs are also found valuable to specify practical languages \cite{JiaKT21deriv,JiaKT23TOPLAS}.}

A language is called a \emph{visibly pushdown language} (VPL) if it
can be generated by a VPG. VPLs enjoy the same appealing theoretical
closure properties as regular languages; e.g., the set of VPLs is
closed under intersection, concatenation, and
complement~\cite{10.1145/1516512.1516518}.  Further, since VPLs are a subset of
deterministic context-free languages, it is always possible to build a
deterministic pushdown automaton from a VPL.

A VPG \cite{10.1145/1516512.1516518} is formally defined as a tuple \((V, \Sigma, P,
L_0)\), where \(V\) is a set of nonterminals, \(\Sigma\) a set of
terminals, \(P\) a set of production rules, and \(L_0 \in V\) the start
nonterminal.
The set of terminals \(\Sigma\) is partitioned into three kinds:
\(\plainsyms\), \(\callsyms\), \(\retsyms\), which contain \emph{plain},
\emph{call}, and \emph{return} symbols, respectively.  The stack
action associated with an input symbol is fully determined by the kind
of the symbol: an action of pushing to the stack is always performed
for a \emph{call symbol}, an action of popping from the stack is
always performed for a \emph{return symbol}, and no stack action is
performed for a \emph{plain symbol}.  Notation-wise, a terminal in
\(\callsyms\) is tagged with \guilsinglleft\ on the left, and a terminal
in \(\retsyms\) is tagged with \guilsinglright\ on the right. For
example, \(\ac\) is a call symbol in \(\callsyms\), and \(\bc\) is a return
symbol in \(\retsyms\).

\emph{Well-matched VPGs} produce strings where each call symbol is
always paired with a return symbol. They are formally defined as
follows:
\begin{definition}[Well-matched VPGs]
A grammar \(G = (V, \Sigma, P, L_0)\) is a well-matched VPG if every production rule in \(P\) adheres to one of the following forms:
\begin{enumerate}
\item \(L\to\epsilon\), where \(\epsilon\) stands for the empty string.
\item \(L\to c L_1\), where \(c\in\plainsyms\).
\item \(L\to \ab{L_1}L_2\), where \(\ac\in\callsyms\) and \(\bc\in\retsyms\).
\end{enumerate}
\end{definition}
Note that in \(L\to c L_1\) terminal \(c\) must be a plain symbol, and in
\(L\to \ab{L_1}L_2\) a call symbol must be matched with a return symbol;
these requirements ensure that any derived string must be
well-matched. This is useful in languages like XML, where tags open and close
in a nested, well-matched manner. For instance, the grammar rule
``$\text{element} \to \text{OpenTag content CloseTag Empty} \mid
\text{SingleTag Empty}$'' represents an XML element that either contains
content within matched open and close tags or is an empty single tag.



In this paper, we consider only well-matched VPGs, and use the term
VPGs for well-matched VPGs.  We {also} call rules in the form
of \(L\to \ab{L_1} L_2\) \emph{matching rules}, and rules of the form
\(L\to c L_1\) \emph{linear rules}.


\subsection{Visibly Pushdown Automata}

\newcommand{\deltacall}{\delta_{\text{call}}}
\newcommand{\deltaret}{\delta_{\text{ret}}}
\newcommand{\deltapln}{\delta_{\text{pln}}}

\edited{
A Visibly Pushdown Automaton (VPA) \cite{10.1145/1007352.1007390} on finite strings over symbols \(\callsyms\), \(\retsyms\), and \(\plainsyms\) is a tuple \(\vpa = (Q, q_0, \Gamma, \delta, Q_F)\) where \(Q\) is a finite set of states, \(q_0 \in Q\) is the initial state, \(\Gamma\) is a finite stack alphabet that contains a special bottom-of-stack symbol \(\bot\), \(\delta = \deltacall \cup \deltaret \cup \deltapln\) is the transition function, where \(\deltacall : Q \times \callsyms \rightarrow Q \times (\Gamma \setminus \{\bot\})\), \(\deltaret : Q \times \retsyms \times \Gamma \rightarrow Q\), and \(\deltapln : Q \times \plainsyms \rightarrow Q\), and \(Q_F \subseteq Q\) is a set of final states.


The function \(\deltacall(q, \ac) = (q', \gamma)\) means that upon reading \(\ac\), state \(q\) is changed to state \(q'\) and \(\gamma\) is pushed onto the stack.
Similarly, \(\deltaret(q, \bc, \gamma) = q'\) means that upon reading \(\bc\) and the stack top is \(\gamma\), \(q\) is changed to \(q'\) and \(\gamma\) is removed from the top of the stack (if \(\gamma\) is \(\bot\), the empty stack remains unaltered).
Finally, \(\deltapln(q, c) = q'\) means state \(q\) is changed to \(q'\) upon reading symbol \(c\).

A \emph{stack} is a non-empty finite sequence over \(\Gamma\) ending in the bottom-of-stack symbol \(\bot\). The set of all stacks is denoted as \(St = (\Gamma \setminus \{\bot\})^* \cdot \{\bot\}\). A \emph{configuration} is a pair \((q, \stack)\) of state \(q\) and stack \(\stack \in St\). We define the \emph{single-step transition} of configurations \(\delta((q, \stack), i)\) as tuple \((q', \stack')\), based on the type of symbol \(i\) and the transition functions:
\begin{enumerate}
  \item If \(i \in \callsyms\), then \((q', \gamma) = \deltacall(q, i) \) and \(\stack' = \gamma \cdot \stack\) for certain \(\gamma \in \Gamma\);
  \item If \(i \in \retsyms\), then \(q' = \deltaret(q, i, \gamma) \) for certain \(\gamma\) and either \(\stack = \gamma \cdot \stack'\), or \(\gamma = \bot\) and \(\stack = \stack' = \bot\);
  \item If \(i \in \plainsyms\), then \(q' = \deltapln(q, i) \) and \(\stack' = \stack\).
\end{enumerate}
We extend the single-step transition for string \(si\) as \(\delta((q, \stack), si)=\delta(\delta((q, \stack), s), i)\).
A string \(s \in \Sigma^*\) is \emph{accepted} by VPA \(\vpa\) if \(\delta((q_0,\bot), s) \in Q_F\). The language of \(\vpa\) is the set of strings accepted by \(\vpa\).
}

\subsection{Angluin's L-Star Algorithm}
We next briefly discuss the \lstar algorithm~\cite{10.1016/0890-5401(87)90052-6}, which learns a finite state automaton from a MAT in polynomial time. We first introduce a notion of equivalence.
We define two strings equivalent w.r.t. to language \(\lang\) 
if extending them with any suffix \(w\) has the same membership result in \(\lang\):
\[
s_1 \eqlstar s_2 \definesym \forall w, s_1w \in \lang \lequiv s_2w \in \lang.
\]

We also introduce a notion of \emph{approximate equivalence} relative to suffixes in a test-string set \(T\):
\[
s_1 \eqlstar_T s_2 \definesym \forall w \in T, s_1w \in \lang \lequiv s_2w \in \lang.
\]

\newcommand{\tbl}{H}

The \lstar algorithm operates in iterations and maintains two sets of
strings: \(Q\) and \(T\), both of which start with \(\set{\epsilon}\).  The
set \(T\) contains a set of test strings.  The set \(Q\) contains a set of
strings that are \emph{separable} by \(T\), which means any two different strings
in \(Q\) are \emph{not} \(T\) equivalent: \(\forall s_1\ s_2 \in Q, s_1 \neq s_2 \limply s_1
\not \simeq_T s_2\). In addition, \((Q,T)\) is \emph{closed} in the sense
that for any \(s \in Q\) and any symbol \(c\) there exists \(s' \in Q\) such
that \(sc \simeq_T s'\).

Given separable and closed \((Q,T)\), we can construct a hypothesis DFA:
each string in \(Q\) becomes a state and we add a transition from \(s \in
Q\) with input symbol \(c\) to the unique state \(s' \in Q\) such that \(sc
\simeq_T s'\); the initial state is the empty string \(\epsilon\) and
acceptance states are those \(Q\) strings that are in \(\lang\).  With the
hypothesis DFA, we can ask the MAT to
check if the DFA is equivalent to the oracle language. If it is, a DFA
for the oracle language has been learned and the algorithm
terminates. If not, the MAT gives a counterexample. With the
counterexample, the algorithm can extend \(Q\) and \(T\) and then use
membership queries provided by the MAT to make \(Q\) and \(T\) separable
and closed again.  Details can be found in the \lstar paper~\cite{10.1016/0890-5401(87)90052-6}. 

\newcommand{\learn}[3]{\mathbf{learn}(#1,#2,#3)}
\newcommand{\consVPA}[1]{\mathbf{constructVPA}(#1)}
\newcommand{\closeVPA}[2]{\mathbf{close}(#1,#2)}
\newcommand{\updateVPA}[3]{\mathbf{update}(#1,#2,#3)}
\newcommand{\tagInfer}[2]{\mathbf{tagInfer}(#1,#2)}
\newcommand{\tagSearchKW}{\mathbf{search}}
\newcommand{\candidateNPKW}{\mathbf{candidateNesting}}

\newcommand{\tagSearch}[1]{\tagSearchKW(#1)}
\newcommand{\candidateNP}[2]{\candidateNPKW(#1,#2)}

\newcommand{\taggedOracleLang}[1]{\hat \lang_{#1}}
\newcommand{\oracleVPL}{\taggedOracleLang{\oracle}}

\newcommand{\ndone}{N_{\mathrm{done}}}

\section{\vstar for a Character-Based VPL} \label{sec:vstarActive}

\vstar learns a Visibly Pushdown Automaton (VPA) using a MAT, which provides both membership and equivalence
queries. For ease of exposition, we divide our discussion into two
steps: in this section, we consider grammar inference for a
character-based VPL, in which the tagging of call/return symbols is on
individual characters; in the next section, we consider grammar
inference for a token-based VPL, a more realistic setting in which the
tagging is performed on tokens (sequences of characters).

For both steps, we develop an algorithm that not only infers the
tagging but also constructs a VPA.  We further prove that this
constructed VPA achieves exact learning, meaning it recognizes the
oracle language.  We also provide an analysis of the algorithm's time
complexity.

This section proceeds as follows: we start with a precise problem
statement in Section~\ref{subsec:charVPL_problem_statement}; we then
{introduce a new algorithm for learning a VPA in
  Section~\ref{subsec:learn_vpa} assuming a given tagging
  function}. Then, in Section~\ref{subsec:learn_tag}, we {study how to
  infer a tagging function that makes the VPA-learning algorithm
  terminate and achieve exact learning}.

\subsection{Problem Statement}\label{subsec:charVPL_problem_statement}
\vstar seeks to infer a Visibly Pushdown Grammar (VPG) from a
black-box oracle that knows a VPL. We next define the precise
knowledge of the oracle and what queries it allows.

We assume \( \Sigma \) is the alphabet set from which valid strings can
draw characters.
A VPL tags each character \( i \) in \( \Sigma \) as a call symbol \(
\callsym{i} \), a return symbol \( \retsym{i} \), or a plain symbol.
This is modeled by a \emph{tagging function} \( t:\Sigma\to\hat\Sigma
\), which maps a character \( i \) to either \( \callsym{i} \), \(
\retsym{i} \), or \( i \) itself. This function extends to strings: 
\( t(s)=t(s[1])\ldots t(s[n]) \), where \( n \) is the length
of \( s \) and \( s[j] \) is its \( j \)-th character.
Given a tagging function \( t \), we define the terminal
set \( \hat\Sigma_t \) (also denoted as \( \hat\Sigma \) when \( t \) is
clear from the context) as the set of tagged characters: \(
\hat\Sigma_t = \callsyms \cup \plainsyms \cup \retsyms \), where \(
\callsyms \), \( \retsyms \), and \( \plainsyms \) include call,
return, and plain symbols defined by \( t \), respectively.


An oracle \( \oracle \) knows {a language \( \lang\subseteq\Sigma^* \) and} a tagging function \( t_\oracle \) such that {the tagged language \( \oracleVPL = \{ t_\oracle(s) \mid s\in \lang \} \) is a}
VPL  over terminal set \( \hat\Sigma_{t_\oracle} \). The oracle \( \oracle \) can
answer membership and equivalence queries in active learning.
The oracle's ability to
answer these queries is modeled as two functions.
The \emph{membership query function}, \( \memquery_\lang: \Sigma^* \to
\{\text{True}, \text{False}\} \), is defined
as follows:
\[
\memquery_\lang(s) =
\begin{cases} 
\text{True} & \text{if } s\in \lang, \\
\text{False} & \text{otherwise}.
\end{cases}
\]
That is, it returns true \emph{iff} the input string \( s \) belongs to \( \lang \). 
Note that input strings to membership queries do not carry tags, which reflects 
the fact that existing oracles are typically recognizers/parsers that take untagged strings.
An example oracle used in our experiments is an off-the-shelf JSON parser, which 
takes untagged JSON strings; the goal of \vstar is to learn the JSON grammar from
this oracle. Also note that we sometimes abuse the notation and pretend that \( \memquery_\lang \) can also take
tagged strings, in which case it performs membership testing using the string
after tagging is removed; i.e., for a tagging function \( t \), \( \memquery_\lang(t(s)) \) is defined as \( \memquery_\lang(s) \).

We next define the \emph{equivalence query function}, which checks the
equivalence between the oracle language \( \lang \) and the language
defined by a hypothesis VPA \( \vpa \) proposed by some learning
algorithm. One complication is that the tagging function produced by
the learner might be different from the oracle tagging
function, even if the underlying untagged language is the same as the
oracle one. This is due to the inherent flexibility of VPL tagging.
As an example, suppose the oracle language is $\{ (\ac
\callsym{g})^k (\retsym{h} \bc)^k \mid k \geq 0 \} $, then its
underlying untagged language is the same as the untagged language of
\( \{ (\ac {g})^k ({h} \bc)^k \mid k \geq 0 \} \), which tags only \( a \) and
\( b \), or of \( \{ (a \callsym{g})^k (\retsym{h} b)^k \mid k \geq 0 \} \),
which tags only \( g \) and \( h \).

Since what is relevant is the underlying untagged language, we should
allow a learner to learn a different tagging function. Therefore we
assume that the learner produces a hypothesis Visibly Pushdown
Automaton (VPA) \( \vpa \), as well as a hypothesis tagging function \(
t_{\vpa} \). The learner should achieve \emph{exact learning},
formally stated as $\forall s \in \Sigma^*,\; \memquery_\lang(s) =
\memquery_{(\vpa,t_{\vpa})}(s)$, where
\[
\memquery_{(\vpa,t_{\vpa})}(s) =
\begin{cases} 
\text{True} & \text{if } t_{\vpa}(s) \text{ is accepted by } \vpa, \\
\text{False} & \text{otherwise}.
\end{cases}
\]
Now the {equivalence query function} \( \mathcal{E} \) is defined
as follows: \( \mathcal{E}(\vpa,t_{\vpa}) \) returns none when the oracle
language is equivalent to the untagged language recognized by \( \vpa \)
and otherwise returns some \( s \) such that 
\( \memquery_\lang(s) \neq \memquery_{(\vpa,t_{\vpa})}(s) \).

\vstar's active learning goal is, with an oracle that provides
membership and equivalence queries, to learn a tagging function \( t \)
and a VPA \( \vpa \) so that exact learning is achieved.

\edited{\paragraph{The Unique Pairing assumption for oracle languages} To simplify the tagging inference algorithm that will be discussed in Section~\ref{subsec:learn_tag}, 
we assume that 
in the oracle VPL \( \hat\lang_\oracle= \{ t_\oracle(s) \mid s\in\lang \} \), a call symbol is uniquely paired with a return symbol; i.e., if \( \ac \) is matched with \( \bc \) in one sentence, then \( \ac \) can be matched with only \( \bc \) in every sentence of the language.
This assumption simplifies our algorithm design, and 
is satisfied by languages we experimented with (e.g., XML and JSON).
We now represent pairs \( (a,b) \) in a tagging function \( t \) as a \emph{tagging} \( T\subseteq 2^{\Sigma\times\Sigma} \), where \( t(a)=\ac \), \( t(b)=\bc \). 
While Algorithm~\ref{alg:backtracking_tagging} technically can be adjusted to operate without the above assumption, efficiency would be significantly decreased.
}

\subsection{Learning VPA with Known Tagging} \label{subsec:learn_vpa}
This subsection outlines an algorithm for learning a Visibly Pushdown
Automaton (VPA) using a MAT, assuming a tagging function \( t \) as input.  
\( \hat \Sigma_t \) is the tagged alphabet according to \( t \) and \( \taggedOracleLang{t} = \{ t(s) \mid s\in\lang \} \) is the 
oracle language \( \lang \) tagged with \( t \). We assume that \( t \) must make
\( \hat \lang \) contain a set of well-matched strings.
To avoid clutter, we will write \( \hat \Sigma \) for \( \hat \Sigma_t \) and \( \hat \lang \) for
\( \taggedOracleLang{t} \) in this subsection.
While there were prior VPA-learning algorithms proposed under this setting, 
some required more information from the oracle, \edited{such as
the stack content during VPA execution}~\cite{michaliszyn_et_al:LIPIcs.MFCS.2022.74}.
\edited{\citet{Isberner2015FoundationsOA}  used advanced discrimination tree structures to minimize the number of membership queries;
however, both \citet{Isberner2015FoundationsOA} and \citet{howar2012active} discussed that discrimination tree-based algorithms could significantly raise the number of equivalence queries. 
Since in our implementation we simulate equivalence queries using membership queries (see Section~\ref{sec:Evaluation}),
increasing the number of equivalence queries would escalate the simulation effort.
}

\edited{In this section, we introduce a VPA learning algorithm based on \(k\)-SEVPA (\cite{Alur2005Congruences,kumar2007visibly}) and demonstrate its polynomial-time efficiency in Theorem~\ref{thm:learn_vpa}. Although the concept of polynomial-time VPA learning has been previously explored, as in \citet{Isberner2015FoundationsOA}'s TTT-VPA, our approach differs by adopting a table-based methodology, inspired by the clarity and directness of the \lstar algorithm~\cite{10.1016/0890-5401(87)90052-6}. 
This shift not only simplifies the presentation but also makes it easy to interface with tag-inference algorithms that we will discuss later in this paper.}

Our algorithm is outlined in Algorithm~\ref{alg:learn_vpa}.
At every iteration, it maintains a set of \emph{separable} and \emph{closed} equivalence states 
and test strings in \( \qc \). The current states are used to produce a hypothesis VPA through \( \consVPA{\qc} \).
It then queries the oracle using an equivalence query. 
If the query does not produce a counterexample, then the iterative process terminates
with the hypothesis VPA as the result; otherwise, the returned counterexample is tagged through
the assumed tagging function \( t \) and employed to refine
the current set of equivalence states and test strings, through \( \updateVPA{-}{-}{-} \) and \( \closeVPA{-}{-} \). Next we describe
\( \consVPA{-} \), \( \closeVPA{-}{-} \), and \( \updateVPA{-}{-}{-} \), starting with some background information.

\begin{algorithm}[t]
  \SetAlgoLined
  \caption{The \(\learn{\oracle}{t}{\hat\Sigma}\) function that learns a VPA from a MAT.}
  \label{alg:learn_vpa}
  
  \KwIn{Oracle \(\oracle\) with membership queries and equivalence queries \( \eqquery \), tagging function \( t \), terminals \( \hat\Sigma=\callsyms\cup\retsyms\cup\plainsyms \).}
  \KwOut{Learned VPA \( \vpa_\qc \).}

  Initialize \( Q_{i,i=0..|\callsyms|} \) as \( \{ \epsilon \} \), \( C_0 \) as \( \{ \epsilon \} \), and \( C_{j,j=1..|\callsyms|} \) as \( \left\{ \left( \ac_{j,j=1..|\callsyms|}, \bc \right) \mid \bc \in \retsyms \right\} \)\;

  \( \qc \leftarrow \closeVPA{\oracle}{\qc} \)\;
  
  \While{\( \eqquery\left(\consVPA{\qc}\right) \) produces a counterexample \( s \)}{
    \( \qc' \gets \updateVPA{\oracle}{\qc}{t(s)} \)\;
    \( \qc \leftarrow \closeVPA{\oracle}{\qc'} \)\;
  }

  \Return{\( \consVPA{\qc} \)};
\end{algorithm}

\subsubsection{Background: \ksevpa and Congruence Relations}
Unlike regular languages, a VPL may not have a \emph{unique}
minimum-state deterministic pushdown recognizer. Nonetheless,
partitioning the call symbols into \( k \) distinct groups and
mandating the following ensure the existence of a unique minimal
VPA: (1) the states are partitioned to a set of \( k+1 \) modules (each is
a set of states), with the \( 0 \)-th module as the base module with
the initial state and the \( i \)-th module for the \( i \)-th group of call symbols
with \( i \in [1..k] \); (2) the machine stays in the same module when encountering a plain
symbol; the machine transitions to a unique entry state in the \( i \)-th
module when encountering a call symbol from the \( i \)-th group; the
machine transitions back to the caller module when encountering a
return symbol. Such a VPA is known as a \( k
\)-Single Entry VPA (\ksevpa~\cite{Alur2005Congruences}) and is similar to
a control-flow graph with the \( 0 \)-th module for the main function and the \( i \)-th
module for the \( i \)-th function in a program.

The minimal \ksevpa can be defined with a set of congruence relations.
\begin{definition}
\label{def:ksevpaCongr}
  [Congruence Relations for the Minimal \ksevpa~\cite{Alur2005Congruences}]
  Given a VPL \(\hat\lang\) over \( \hat\Sigma \), let \( \callsyms^i \), \( i \in [1..k] \), represent the \( i \)-th group of call symbols. Given \emph{well-matched} strings \( s_1 \) and \( s_2 \), we introduce \(k+1\) congruence relations:
  \begin{eqnarray}
    s_1 \sim_0 s_2 \text{ iff }&
      \forall w \in \hat\Sigma^*,\ s_1w \in \hat\lang \iff s_2w \in \hat\lang; \\
    s_1 \sim_i s_2 \text{ iff }&
      \forall w, w' \in \hat\Sigma^*,\ \forall \ac \in \callsyms^i,\ w\ac s_1w' \in \hat\lang \iff w \ac s_2 w'\in \hat\lang, \text{ for } i \in [1..k].
  \end{eqnarray}
\end{definition}
Note that \(\sim_0\) is the Myhill-Nerode right congruence
and can be used to construct the minimal DFA for a regular language.
For \(\sim_i\) when \( i \in [1..k] \), the context
strings assume specialized forms: the left context string ends with a
call symbol from the \(i\)-th group, and the string \( w\ac
w' \) is a well-matched string, since both \( w \ac s_1 w'\in \hat\lang \) and
\(s_1\) are well matched.
From the above congruence relations, we can construct the minimal \ksevpa: 
the equivalence classes of \( \sim_i \) become the states of
the \( i \)-th module, with \( [\epsilon]_{\sim_i} \) being the unique entry
  state of the \( i \)-th module; transition edges can also be
added (see \citet{Alur2005Congruences}): e.g., \( [s]_{\sim_i} \) transitions to \( [si]_{\sim_i} \)
for plain symbol \( i \), and to \( [\epsilon]_{\sim_j} \) for call symbol \( \ac \in \callsyms^j \).

In our algorithm, we set \(k\) to be the number of call symbols decided by the input tagging function 
so that each call symbol is in its own group. We write \( \ac_i \) for the \( i \)-th call symbol. This partitioning is practical because call symbols often fulfill diverse roles and hence find themselves in separate contexts. Further, Proposition 2 in \citet{Alur2005Congruences} tells that enlarging \(k\) can lead to a more compact VPA.

\subsubsection{Access Words and Test Words}
At each step, \vstar maintains \( \qc \), which contains
\begin{enumerate}
  \item a set of \( k+1 \) modules \( Q_0 \) to \( Q_k \), each containing empty string \( \epsilon \) and
  a set of well-matched \emph{access words} in \( \hat\Sigma^* \),
  \item and a set of \emph{test words} \( C_0 \) to \( C_k \), with \( C_0 \) containing strings in \( \hat\Sigma^* \) for testing \( Q_0 \) and
  \( C_i \) containing strings in the form of \( (w\ac_i,w') \) for testing \( Q_i \), where \( w \) and \( w' \) are in \( \hat\Sigma^* \), \( \ac_i \) is the \( i \)-th call symbol, and \( w \ac_i w' \) is well matched.
\end{enumerate}

Given test words \( C_{i\in[0..k]} \), two well-matched strings \( q_1 \) and \( q_2 \) are \emph{\( C_i- \)equivalent}, denoted as \( q_1 \sim_{C_i} q_2 \), if (1) when \( i = 0 \), \( \forall w \in C_0 \), \( q_1w\in \hat\lang \) iff \( q_2w\in \hat\lang \); and (2) when \( i \in [1..k] \), \( \forall (w\ac, w') \in C_i \), \( w \ac q_1w'\in \hat\lang \) iff \( w \ac q_2w'\in \hat\lang \). These are essentially the same equivalence relations as those in Definition~\ref{def:ksevpaCongr}, relative to test words in \( C_i \). 

We define the following two properties of \( \qc= \predset{(Q_i,C_i)}{i \in [0..k]} \):
\begin{enumerate}
  \item \textbf{Separability}: no two distinct strings in \( Q_i \) are \( C_i \)-equivalent, meaning \( \forall q\ q' \in Q_i, q \neq q' \limply q \not\sim_{C_i} q' \).
  \item \textbf{Closedness}: for every \( q\in Q_i \) and \( m\in\Sigma_M \) (defined below), there is some \( q' \in Q_i \) such that \( qm \sim_{C_i} q' \).
\end{enumerate}

\begin{definition}[Nested Words and \( \Sigma_M \)]
  Given \( \hat\Sigma = \callsyms \cup \plainsyms \cup \retsyms \), along with \( (Q_i, C_i)_{i\in[1..k]} \), we define the \emph{nested words} for \( (Q_i, C_i) \), denoted as \( M_i \), as
  \[ M_i = \{ \ac_i q \bc \mid q \in Q_i, \bc \in \retsyms \}, \]
  where \( \ac_i \) is the \( i \)-th call symbol.
  We define \( \Sigma_M = \cup_{i} M_i \cup \hat\Sigma \).
\end{definition}

Our learning algorithm is then based on the following set of propositions.

\begin{definition}[\( \consVPA{\qc} \) function] \label{def:consVPA}
For separable and closed \( \qc= \predset{(Q_i,C_i)}{i \in [0..k]} \), we can construct a hypothesis \ksevpa, denoted as \( \vpa \) as follows.
{The set of states of \( \automaton \) is \( \bigcup_{i\in[0..k]} Q_i \). We write \( q\in Q_i \) as \( [q]_{i} \). The initial state is \( [\epsilon]_0 \). Define the set of acceptance states, \( Q_F \), to be \( \{ [q]_0 \mid q \in Q_0 \cap \hat\lang \} \), which can be constructed via membership queries. The transition function \( \delta \) from the current state \( [q]_i \), \( i \in [0..k] \),  and the next input symbol is defined as follows:
    \begin{enumerate}
    \item For plain symbol \( c \), the transition is \( [q]_{i} \xrightarrow{c} [q']_{i} \), where \( q'\sim_{C_i}qc \).
    \item For call symbol \( \ac_{j} \), the transition is \( [q]_{i} \xrightarrow{\ac_{j}, \text{ push } ([q]_{i},\ac_{j})} [\epsilon]_{j} \), the unique entry state for module \( j \).
    \item For return symbol \( \bc \), the transition is \( [q]_{i} \xrightarrow{\bc, \text{ pop } ([q']_{j},\ac_{i})} [q'']_{j} \), where \( q'' \sim_{C_j}q'\ac_{i} q \bc \).
  \end{enumerate}
  The target state in each transition exists by closedness and is unique by separability.
  To run the VPA on a string \( s \), start with the initial state \( [\epsilon]_0 \) and an empty stack and use \( \delta \) for transitions. The automaton accepts a string if it terminates in a configuration
with a state within \( Q_F \) and an empty stack.
}
\end{definition}

\begin{proposition}\label{prop:bound_states}
  If \( \qc= \predset{(Q_i,C_i)}{i \in [0..k]} \) is separable {and language \( \hat\lang=\{ t(s) \mid s\in\lang \} \) is a VPL}, then the number of states in \( \consVPA{\qc} \) 
is bounded above by the number of states in the minimal \( k \)-SEVPA for VPL \( \hat\lang \).
\end{proposition}
\begin{proof}
  For two strings \( s_1 \) and \( s_2 \), if \( s_1 \sim_i s_2 \) (Definition~\ref{def:ksevpaCongr}), then \( s_1 \sim_{C_i}s_2 \). Hence, the number of equivalence classes of \( \sim_{C_i} \) is less than that of \( \sim_i \), which corresponds to the number of states of the \( i \)-th module in the minimal \( k \)-SEVPA. Further, since \( Q_i \) is separable, each element of \( Q_i \) corresponds to a unique equivalence class of \( \sim_{C_i} \). Therefore, \( |Q_i| \) is bounded above by the number of equivalence class of \( \sim_{C_i} \), which is bounded
above by the number of states of the \( i \)-th module in the minimal \ksevpa.
Since the number of states in \( \consVPA{\qc} \)  is \( \sum_{i\in [0..k]}|Q_i| \), it
is upper bounded by the number of states in the minimal \ksevpa for \( \hat\lang \).
\end{proof}

\begin{proposition}\label{prop:close_vpa}
  If \( \qc= \predset{(Q_i,C_i)}{i \in [0..k]} \) is separable but not closed, then using membership queries one can find \( i \) and \( q \in \hat\Sigma^* \setminus Q_i \) such that \( (Q_i \cup \{q_i\}, C_i) \) and the rest \( (Q_j,C_j)_{j\neq i} \) remain separable.
\end{proposition}
\begin{proof}
  Since \( (Q_i, C_i)_{i\in[0..k]} \) are not closed, there exists \( q \in Q_i \) for certain \( i \) and \( m \in \Sigma_M \) such that \( qm \) is not \( C_i \)-equivalent to any state in \( Q_i \). Using membership queries (by enumerating all test strings in \( C_i \))
we can find such a \( q \) and \( m \), and then add \( qm \) to \( Q_i \), which remains separable by construction.
\end{proof}
Algorithm~\ref{alg:close_vpa} outlines the \( \closeVPA{\oracle}{\qc} \) function,
which keeps applying Proposition~\ref{prop:close_vpa} until \( \qc \)
becomes separable and closed. 

\begin{algorithm}[t]
  \SetAlgoLined
  \caption{The \( \closeVPA{\oracle}{\qc} \) function.}
  \label{alg:close_vpa}
  
  \KwIn{Oracle \( \oracle \) and separable \( \qc= \predset{(Q_i,C_i)}{i \in [0..k]} \).}
  \KwOut{Separable and closed \( \qc' \).}

  Initialize \(\Sigma_M \) as \( \bigcup_{i=1..k} \{ \ac_i q \bc \mid q\in Q_i, \bc\in\retsyms \} \cup \hat\Sigma\)\;
  Initialize the work list \(W\) as \(\{ (q,i,m) \mid q\in Q_i,i \in 0..k, m \in \Sigma_M\}\)\;

  \While{\(W\) is not empty}{
    Take \((q,i,m)\) from \(W\)\;
    \If{ \(\forall q'\in Q_i\), \(qm\not\sim_{C_i}q'\) }{
      \( Q_i \gets Q_i \cup \set{qm} \)\;
      \( W \gets W \cup \predset{(qm, i, m')}{m' \in \Sigma_M} \)\;
      \If{\(i>0\)}{
        \( \Sigma_M\gets \Sigma_M\cup \{ \ac_iqm\bc \mid \bc\in \retsyms \} \)\;
        \( W \gets W \cup \predset{(q'', j, \ac_iqm\bc)}{q'' \in Q_j, j \in [0..k],\bc \in \retsyms} \)\;
      }
    }
  }
  \Return{\( \predset{(Q_i,C_i)}{i \in [0..k]} \)}\;
\end{algorithm}

\begin{proposition}\label{prop:find_counterexample}
  Suppose that \( \qc= \predset{(Q_i,C_i)}{i \in [0..k]} \) is separable and closed, and let \( \vpa \) be the hypothesis automaton (Definition~\ref{def:consVPA}). Suppose the oracle returns a counterexample \( s \) for 
an equivalence query with \( \vpa \). Using \( \log |s| \) membership queries, one can find \( i \in [1..k] \) and \( q \in \hat\Sigma^* \setminus Q_i \) and \( (w\ac_i,w') \in \hat\Sigma^*\callsyms\times\hat\Sigma^* \) such that \( (Q_i \cup \{q\}, C_i \cup \{(w\ac_i,w')\}) \) is separable, or  find \( w\in\hat\Sigma^* \) when \( i=0 \) such that  \( (Q_0 \cup \{q\}, C_0\cup \{w\}) \) is separable. 
\end{proposition}
\begin{proof}  
  Let \( n \) be the length of \( s \). Let \( q_0 = [\epsilon]_0 \) be the initial state of \( \vpa \), and \( \delta \) be the transition function of \( \vpa \). For \( i = 1, \ldots, n \), define \( q_i = \delta(q_0, s[1] \ldots s[i]) \) to be the state in \( Q_j \) reached by \( \vpa \) after reading the prefix \( s[1] \ldots s[i] \) of \( s \), and define \( T_i \) as the corresponding stack. For convenience, we write \( [q_i]_j \) for the state \( q_i \) in module \( j \).

  We define the context of \( q_i \) as follows. When \( T_i \) is empty, we define the context to be \( (\epsilon, s[i+1] \ldots s[n]) \). Otherwise, let \( T_i=(q_{j_{n'}},\ac_{j_{n'}})\cdots (q_{j_1}, \ac_{j_1})\cdot\bot \).
  We define the context as 
  \[ (q_{j_1}\ac_{j_1}\dots q_{j_{n'}}\ac_{j_{n'}}, s[i+1] \ldots s[n]). \]
  We denote the context of \( q_i \) as \( (w_i,w_i') \).
  We say that state \( q_i \) is correct if \( \memquery_\lang(w_iq_iw_i') = \memquery_\lang(s) \). 

  State \( q_0 = [\epsilon]_0 \) is obviously correct since its context is \( (\epsilon, s) \). However, state \( q_n \) must be incorrect because of the following.
  First, state \( q_n \) must be in module \( 0 \): if the counterexample \( s \) is accepted by \( \vpa \), then \( s \) is well-matched under \( t \); otherwise, the counterexample \( s \) is in \( \hat \lang \), which is a well-matched language. Therefore, we write \( [q_n]_0 \).
  Next,   since \( s \) is a counterexample, we have \( \memquery_{(\vpa,t)}(s) \neq \memquery_\lang(s) \). By the construction of \( \vpa \), we have \( \memquery_{(\vpa,t)}(s)=\memquery(q_n) \). Therefore, we have \( \memquery(q_n)\neq \memquery_\lang(s) \), which means state \( [q_n]_0 \) is incorrect.

  Now we can then use binary search (using \( \log |s| \) membership queries)  to find \( i \) such that \( [q_i]_j \) is correct, 
  while \( [q_{i+1}]_{j'} \) is incorrect.
  We first show that \( s[i+1] \) cannot be a call symbol. Otherwise, we must have \( q_{i+1}=\epsilon \)
  and \( T_{i+1}=(q_{i},s[i+1])\cdots (q_{j_1}, \ac_{j_1})\cdot\bot \).
  The context of \( q_{i+1} \) is \( (q_{j_1}\ac_{j_1}\dots q_{i}s[i+1], s[i+2] \ldots s[n]) \).
  We have 
  \[ \memquery_\lang(w_{i+1}q_{i+1}w_{i+1}')=\memquery_\lang(q_{j_1}\ac_{j_1}\dots q_{i}s[i+1] s[i+2] \ldots s[n])=\memquery_\lang(w_{i}q_{i}w_{i}')=\memquery_\lang(s), \]
  but \( \memquery_\lang(w_{i+1}q_{i+1}w_{i+1}')\neq\memquery_\lang(s) \), a contradiction.

  Assume \( s[i+1] \) is a plain symbol. Since \( q_i \) is a state in module \( j \), we let \( Q_j' = Q_j \cup \{q_{i}s[i+1]\} \) 
  and \( C_j' = C_j \cup \{(w_{i+1},w_{i+1}')\} \).
  By definition of the transition function of \( \vpa \), \( q_{i+1} \) is the unique element of \( Q_j \) that is \( C_j \)-equivalent to \( q_{i}s[i+1] \). On the other hand, the test \( (w_{i+1},w_{i+1}') \) distinguishes \( q_{i+1} \) from \( q_{i}s[i+1] \), 
since we can get 
\( \memquery_\lang(w_{i}q_iw'_{i}) = \memquery_\lang(w_{i+1}q_is[i+1]s[i+2]...s[n]) \neq \memquery_\lang(w_{i+1}q_{i+1}s[i+2]...s[n]) \).

  Otherwise, \( s[i+1] \) is a return symbol \( \bc \). 
  Let \( T_{i}=(q_{i'},\ac_j)\cdots (q_{i_1}, \ac_{j_1})\cdot\bot \).
  Recall that at state \( q_i \), \( \vpa \) reads \( \bc \) and transfers to \( [q_{i+1}]_{j'} \) such that \( q_{i+1} \sim_{C_{j'}} q_{i'}\ac_jq_i \bc \).
  Notice that \( \memquery_\lang (w_{i+1}q_{i'}\ac_jq_i \bc w_{i+1}')=\memquery_\lang (w_{i} q_i  w_{i}')=\memquery_\lang(s) \). Therefore, \( q_{i'}\ac_jq_i \bc \) is correct.
  We let \( Q_{j'}'=Q_{j'}\cup \{ q_{i'}\ac_jq_i \bc \} \) 
  and \( C_{j'}' = C_{j'} \cup \{(w_{i+1},w_{i+1}')\} \).
  By definition of the transition function of \( \vpa \), \( q_{i+1} \) is the unique element of \( Q_{j'} \) that is \( C_{j'} \)-equivalent to \( q_{i'}\ac_jq_i \bc \). On the other hand, the test \( (w_{i+1},w_{i+1}') \) distinguishes \( q_{i+1} \) from \( q_{i'}\ac_jq_i \bc \). We conclude that \( q_{i'}\ac_jq_i \bc \notin Q_{j'} \), and that \( (Q_{j'}', C_{j'}') \) is separable. 
\end{proof}

We call the procedure in Proposition~\ref{prop:find_counterexample}
\( \updateVPA{\oracle}{\qc}{t(s)} \), which takes a separable and closed \( \qc \) and a
counterexample \( s \) and returns a separable \( \qc' \).

With these lemmas, we can prove the following theorem; its proof is provided in Appendix~\ref{app:thm:learn_vpa}.
\begin{theorem}\label{thm:learn_vpa}
  {Given a tagging function \( t \) such that language \( \hat\lang=\{ t(s) \mid s\in\lang \} \) is a VPL,} the minimal
  \( k \)-SEVPA of language \( \hat\lang \) can be learned in polynomial numbers of equivalence and membership queries.
\end{theorem}

Therefore, if language \( \hat \lang = \{ t(s) \mid s\in\lang \} \) is a VPL, then we can use Algorithm~\ref{alg:learn_vpa} to learn the minimal \( k \)-SEVPA of language \( \hat\lang \). However, in general, a tagging function \( t \) does not necessarily introduce a VPL \( \hat \lang \), even if each sentence in \( \hat\lang \) is well-matched. For example, consider the language \( \hat\lang_\oracle=\{ a^kb^k\mid k>0 \} \) and the tagging \( t \) that maps \( a \) and \( b \) to plain symbols. The resulting language is trivially well matched (as it does not have call/return symbols), but it is not a VPL.
In the next section, we discuss the procedure to find a right tagging function that makes \( \lang \) a VPL.


\subsection{Tagging Inference}\label{subsec:learn_tag}

\begin{figure}
    \centering
    \begin{align*}
        L &\to \ab{\ A\ }\ L \mid c\ B \mid \epsilon\\
        A &\to \callsym{g}\ L\ \retsym{h}\ E\\
        B &\to d \ L \\
        E &\to\epsilon\\[0.5ex]
        \mbox{Seed strings } & S=\set{agcdcdhbcd}
    \end{align*}
    \caption{An oracle VPG and a set of seed strings.}
    \label{fig:example_vpg}
    \Description{An oracle VPG and a seed string.}
\end{figure}

We next propose an algorithm that infers a tagging function \( t \) so that
the tagged language \( \hat\lang_t=\{ t(s) \mid s\in\lang \} \) is a VPL.
Then with Theorem~\ref{thm:learn_vpa}, we can use the tagging function
in Algorithm~\ref{alg:learn_vpa} to learn the VPL efficiently.
Our algorithm takes a set of seed strings \( S \) for inference.
In practice, the seed strings can be collected via existing corpora of data (e.g., a corpus of XML strings) or
via valid inputs to black-box implementations of oracles (e.g., an XML parser).

We will use the oracle VPG 
in Figure~\ref{fig:example_vpg} as a running example, which includes a single seed
string. Note that seed strings are untagged: it is the
task of our algorithm to infer the tagging.
As discussed earlier, the inferred and the oracle tagging
functions may differ. For the example VPG in
Figure~\ref{fig:example_vpg}, we can remove the tags on either
\( (\callsym{a}, \retsym{b}) \) or \( (\callsym{g}, \retsym{h}) \) and the
resulting grammar is still a VPG and generates the same untagged
language. As we will explain, in such a case the algorithm opts for the
outermost tags in its inferred VPG (i.e., \( \callsym{a} \) and
\( \retsym{b} \)), while treating \( {g} \) and \( {h} \) as plain symbols.

This section unfolds as follows. We first introduce a VPL pumping
lemma, which enables a nesting test to filter out invalid taggings. We
then present a tagging inference algorithm based on the nesting test
and state theorems discussing properties of the algorithm. We leave
most of the proofs of these theorems to
Appendix~\ref{app:proofs_of_learn_tag}.

One straightforward observation is that a tagging \( T \) is invalid if some seed string after tagging through \( T \) is not well matched.
However, this test alone would not eliminate too many possibilities. We next introduce a nesting test 
to filter out more invalid taggings. For this, we first propose a pumping lemma for VPLs. This lemma diverges from traditional pumping lemmas for regular languages and context-free languages by focusing on the unique requirements of call and return symbols in VPLs.

\begin{lemma}
  [Pumping Lemma for VPLs]\label{lemma:pumping}
  For any VPL \( \hat\lang \), there exists a positive number \( l \) such that, for any string \( s \) in \( \hat\lang \) with length greater than \( l \), it is possible to express \( s \) according to one of the following conditions:
  \begin{enumerate}
    \item (Regular Pumping) We can partition \( s \) into \( s = uxv \) for strings \( u, x, \) and \( v \), with \( x \) being non-empty, such that \( ux^kv \) remains in \( \hat\lang \) for all \( k \geq 0 \). 
    \item (Nesting Pumping) We can partition \( s \) into \( s = uxzyv \) for strings \( u, x, z, y, \) and \( v \), with \( x \) and \( y \) being non-empty, \( x \) containing a call symbol, and \( y \) containing a return symbol, such that \( ux^kzy^kv \) is valid for all \( k \geq 1 \).
  \end{enumerate}
\end{lemma}

For example, consider the VPG in Figure~\ref{fig:example_vpg}.
Any string in \( \hat\lang \) with length greater than 6 can be decomposed based on the above two ways; e.g., for the tagged seed string \( s=\ac \callsym{g} cdcd \retsym{h} \bc cd \), we have \( (\ac \callsym{g})^k cdcd (\retsym{h} \bc)^k cd \) in the language, for \( k \geq 1 \); it happens it can also be decomposed through regular pumping: we have \( \ac \callsym{g} cd(cd)^j \retsym{h} \bc cd  \) in the language, for \( j \geq 0 \). 
We now extend the concept of nesting pumping to untagged strings, calling them nesting patterns.

\begin{definition}
  [Nesting Patterns]
  For an untagged string \( s \) in the oracle language \( \lang \),
  a nesting pattern is a partitioning of
  \( s = uxzyv \), where (1) \( x \) and \( y \) are non-empty, (2) \( ux^kzy^kv \in \lang \) for all \( k \geq 1 \), (3) but for \( k \neq j \) (both \( \geq 0 \)), \( ux^kzy^jv \not\in \lang \). 
The third condition precludes the possibility that \( uxzyv \)
represents a regular pumping, which allows \(ux^jzyv\) and \(uxzy^jv\)
for all \( j \geq 0 \).  When \( u \), \( z \), and \( v \)
are not the focus, we may succinctly write a nesting pattern in a string as
a pair \((x,y)\).
\end{definition}

\begin{definition}[Compatible Tagging]
  We say that tagging \( t \) is {compatible with a nesting pattern \(s=uxzyv\),} if there exists a pair \((\ac,\bc)\) in \( t \), such that (1) \( x \) includes \( a \) and \( y \) includes \( b \), and (2) {\( t(x) \) includes an unmatched \( \ac \) and \( t(y) \) includes an unmatched \( \bc \)} in \( t(s) \).

  We say that tagging \( t \) is compatible with a set of seed strings \( S \), if (1) strings in \( S \) are well-matched under \( t \), and (2) \( t \) is compatible with all nesting patterns of \( S \). \edited{We say that tagging \( t \) is compatible with language \( \lang \) if it is compatible with each string in the language.}
\end{definition}

\begin{theorem}\label{thm:compatible_tagging}
  Given an oracle language \( \lang \),
  for any tagging \( t \) compatible with \( \lang \), language \( \hat\lang_t=\{ t(s) \mid s\in\lang \} \) is a VPL.
\end{theorem}

For the example in Figure~\ref{fig:example_vpg}, the single seed string's nesting patterns include 
\[ \{(ag, hb),(agcd,hb), \dots, (ag,cdcdhbcd)\}. \] 
One compatible tagging is 
\( \{ (a, b)\} \): firstly, the tagging would make the seed string well matched; secondly, each nesting pattern includes \( (a, b) \). 
By Theorem~\ref{thm:compatible_tagging}, when we tag \( a \) as a call symbol and \( b \) as a return symbol, the oracle language becomes a VPL. 
Other compatible taggings include \( \{ (a, h)\} \), \( \{ (g, h)\} \), \( \{ (g, b)\} \), and \( \{ (a, b), (g, h) \} \). In contrast, the tagging \( \{ (a,h), (g,b) \} \) is incompatible since this tagging would not make the seed string well matched.

Recall that Theorem~\ref{thm:learn_vpa} tells us that, if a tagging
\( t \) makes \( \lang \) a VPL, we can efficiently learn the VPL under active
learning through Algorithm~\ref{alg:learn_vpa}. Now
Theorem~\ref{thm:compatible_tagging} tells us that a compatible
tagging \( t \) makes \( \lang \) a VPL. Therefore, what is remaining is to infer
a compatible tagging. With such a tagging, we can use it
in Algorithm~\ref{alg:learn_vpa} to learn a VPL whose untagged strings
are the same as the oracle language.

\begin{algorithm}[t]
  \SetAlgoLined
  \caption{The \( \tagInfer{\oracle}{S} \) algorithm that infers tagging. 
  }
  \label{alg:backtracking_tagging}

  \KwIn{Oracle \( \oracle \) and seed strings \( S \).}
  \KwOut{Some tagging \( T \) compatible with \( S \), or \texttt{None} if no compatible tagging is found.}

\SetKwProg{Fn}{Function}{:}{}

\SetKwFunction{GetCandidates}{\( \candidateNPKW \)}
\Fn{\GetCandidates{\( S,K \)}}{
    \( N_{S,K} \gets \emptyset \)\;
    \ForEach{{partitioning } \( uxzyv \) {of\ } \( s \in S \)} {
      \lIf{\( \forall k \leq K, \memquery_\lang(ux^kzy^kv)=\mathrm{True} \) 
      {and \( \forall k,j \leq K, k\neq j \limply  \memquery_\lang(ux^kzy^jv)=\mathrm{False} \)}}{\( N_{S,K} \gets N_{S,K} \cup \set{uxzyv} \)}
    }
    \Return{\( N_{S,K} \)}\;
}

\SetKwFunction{TagSearch}{\( \tagSearchKW \)}
\Fn{\TagSearch{\( N, \ndone,T \)}}{
  \lIf{\( N \) is empty}{\Return{\( \texttt{Some}(T) \)}}
  Take a nesting pattern \( uxzyv \) from \( N \)\;
  \If{\( T \) is incompatible with \( uxzyv \)}{
    \ForEach{character \( a \) in \( x \) and \( b \) in \( y \) 
    }{
      \If{all strings in \( S \) are well-matched under \( T \cup \set{(a,b)} \) 
      \\\quad {and \( T \cup \set{(a,b)} \) is compatible with \( \ndone \cup\set{uxzyv} \)}}{
        \( T' \gets \tagSearch{N\setminus \set{uxzyv}, \ndone \cup \set{uxzyv}, T \cup \set{(a,b)}} \)\;
        \lIf{\( T' \) is not \texttt{None}}{\Return{\( T' \)}}
      }
    }
    \Return{\texttt{None}}; \tcp{No compatible tagging found}
  }
  \lElse {\Return{\( \tagSearch{N\setminus \set{uxzyv}, \ndone \cup \set{uxzyv}, T } \)}}
}
Initialize \( K \) as \( 1 \)\;
\Repeat{\( T\neq\texttt{None} \)}{
  \( K \gets K+1 \); \( N_{S,K} \gets  \candidateNP{S}{K} \);
  \( T \gets  \tagSearch{N_{S,K},\emptyset,\emptyset} \);
}
{\Return{\( T \)}\;}
\end{algorithm}

We next describe Algorithm~\ref{alg:backtracking_tagging}, which
performs inference of a compatible tagging using an input set of seed strings. 
Its runtime complexity is exponential in the worst case; however, as will be discussed in our evaluation section, its runtime performance on practical grammars is efficient.
As an overview, starting with a bound \( K=2 \), the algorithm (1) employs
a bounded checking approach in the \( \candidateNPKW \) function to
compute candidate nesting patterns \( N_{S,K} \) for seed strings \( S
\), and (2) for \( N_{S,K} \), the \( \tagSearchKW \) function tries to
find a compatible tagging using a search algorithm (which may
backtrack). If the \( \tagSearchKW \) function fails to find a compatible tagging,
we increase \( K \) by
\( 1 \) and start anew.
In more detail, in \( \candidateNPKW \), for each disjoint substring pair \( (x,y) \) in each seed string, we check if \( ux^kzy^kv \in \lang \) for \( k \leq K \) and check if \( ux^kzy^jv \not\in \lang \) for \( k,j\leq K\) {and} \( k \neq j \). If so, \( (x,y) \) is a candidate nesting pattern. In the \( \tagSearchKW \) function, we begin with an empty tagging \( T \), which tags each character as a plain symbol. We then check if a candidate nesting pattern is already covered by the current tagging; if not, we treat a symbol in \( x \) as a call symbol and a symbol in \( y \) as a return symbol and continue the search process. 

Returning to our example shown in Figure~\ref{fig:example_vpg}, the seed string includes the nesting pattern \( \{(ag,hb)\} \). Our algorithm prioritizes the outermost characters for pairing. Consequently, the pair \( (a,b) \) is selected and it covers all nesting patterns of the seed string, resulting in a compatible tagging \( \{ (a,b) \} \).

The following theorem states that for some bounded \( K \)
Algorithm~\ref{alg:backtracking_tagging} terminates and returns a compatible tagging.

\begin{theorem}[Termination and Correctness of Algorithm~\ref{alg:backtracking_tagging}]
  \label{lemma:backtrack_result}
  Let \( m \) be the number of states of the minimal \( k \)-SEVPA for the oracle VPL.
  There exists a number \( K\leq ((m^2+2m)^2+1)^2 \), such that \( \tagInfer{\oracle}{S} \) returns a tagging that is compatible with a finite set of seed strings \( S \).
\end{theorem}

Note that the theorem is with respect to a finite set of seed strings \( S \). It does not say whether the found tagging is 
compatible with all strings in the oracle language \( \lang \). We address this by demonstrating the existence of a finite set of seed strings \( S_0 \), for which a compatible tagging \( T \) with \( S_0 \) ensures compatibility with all strings in \( \lang \).

\begin{theorem}[Finite and Sufficient Seed Strings]\label{thm:construct_seed_strings}
  For any given oracle language \(\lang\), there exists a finite set of seed strings, denoted as \( S_0 \), such that any tagging that is compatible with \( S_0 \) is also compatible with \(\lang\).
\end{theorem}

The proof of Theorem~\ref{thm:construct_seed_strings} is provided in
Appendix~\ref{app:thm:construct_seed_strings}.  We illustrate how
\( S_0 \) is constructed using the VPG shown in
Figure~\ref{fig:example_vpg}. In general, for each matching rule \(
L\to \ab{A}{B} \) where the nonterminal \( A \) can be recursively
rewritten into \( L \) via a set of derivations, we generate a string
reflecting this recursion and incorporate it into \( S_0 \). The example
VPG includes matching rules \( L \to \ab{A}{L} \) and \( A \to
\callsym{g} L \retsym{h} E \). We start by expanding \( L \) to reveal
\( A \), followed by expanding \( A \) to unveil \( L \), resulting in
the pattern \( \ab{\callsym{g}L\retsym{h}}L \). Then the expansion of
\( L \) to \( cd \) gives us a seed string \( agcdhbcd \). {We also generate
 a seed string witnessing the recursive transition from \( A \)
to \( L \) and back to \( A \), which would lead to strings like 
\( agagcdhbhbcd \). }

In conclusion, it is established that a finite set of seed strings exists, enabling our algorithm to identify a tagging \( T \) that is compatible with the oracle language \( \lang \). This compatibility guarantees that the tagged oracle language is a VPL. With the learned tagging as an input, Algorithm~\ref{alg:learn_vpa} can be employed to achieve exact learning.



\newcommand{\tokenInfer}{\mathbf{tokenInfer}}
\newcommand{\tokenit}{\mathbf{tokenize}}
\newcommand{\validate}{\mathbf{validate}}
\newcommand{\tokenSearch}{\mathbf{tokenSearch}}

\section{\vstar for a Token-Basd VPL}\label{sec:learn_token_tag}

In Section~\ref{sec:vstarActive}, we assumed that tagging of the oracle language is on individual characters; i.e.,
each character is uniquely tagged. This assumption does not always align with real-world scenarios. For instance, in JSON, a curly bracket \texttt{\{} may serve as a call symbol, yet it can also be a key, exemplified by \texttt{\{"\{":true\}}; in XML documents, an opening tag such as \texttt{<p>} functions as a call symbol, but it is composed of multiple characters.
In this section, we enhance \vstar to address these scenarios. 

\subsection{Problem Statement}
The oracle language \( \lang \) is a VPL when sentences in \( \lang \) are converted to sequences of tokens determined by an oracle tokenizer. 
Formally, a tokenizer function, \( \tk: \Sigma^* \rightarrow H_\tk^* \), translates a string \( s \) from \( \lang \) into a sequence of tagged tokens, where \( H_\tk=H_{\text{call}}\cup H_{\text{plain}}\cup H_{\text{ret}} \) represents the set of tagged tokens given by \( \tk \); we write \( H \) when \( \tk \) is clear from the context. 
The language \( \{ \tk(s) \mid s\in\lang \} \) is assumed to be a VPL over tokens in \( H \). Each category of token \( h \) is defined as a regular language, often specified by a regular expression. The notation \( s\in h \) indicates that string \( s \) belongs to token \( h \). We use metasymbols \(\ta \), \(h_c\), or \(\tb \) for call, plain, or return tokens, respectively. 

\begin{figure}[t]
  \centering
    \centering
    \begin{align*}
      L & \to \ {\texttt{OPEN}}\ L\ {\texttt{CLOSE}} \mid \texttt{TEXT} \\
      \texttt{OPEN} & \to \texttt{<p>}\\
      \texttt{CLOSE} & \to \texttt{</p>}\\
      \texttt{TEXT} & \to \text{[a..z]+}
    \end{align*}
    \caption{An example XML grammar and the associated lexical rules.}
    \Description{An example XML grammar and the associated lexical rules.}
    \label{fig:example_xml}
\end{figure}

A toy XML grammar is shown in Figure~\ref{fig:example_xml} and we use the seed string \( s = \texttt{<p><p>p</p></p>} \) as an example. 
The tokens are \texttt{OPEN}, \texttt{TEXT}, and \texttt{CLOSE}. 
The oracle tokenizer converts \( \texttt{<p><p>p</p></p>} \) into the token sequence \texttt{[OPEN,OPEN,TEXT,CLOSE,CLOSE]}, where \texttt{OPEN} is a call symbol, \texttt{TEXT} is a plain symbol, and \texttt{CLOSE} is a return symbol. 

The oracle still provides membership and equivalence queries. The
membership query function \(\memquery_\lang:\Sigma^*\to\{\true,\false\} \) is as before.
However, we change the form of equivalence queries. The reason for the change 
is to convert the oracle language to a character-based VPL so that 
we can reuse Algorithm~\ref{alg:learn_vpa} for learning a hypothesis VPA. 


To model equivalence queries, we first define a converter function.
A tokenizer \( \tk \)
identifies boundaries of call and return tokens for a string. We then use 
\( \convert_\tk:\Sigma^*\to\tilde{\Sigma}_{\tk}^* \) to transform a valid string \( s \in \lang \) into a new string \( \tilde{s}=\convert_\tk(s) \) by inserting artificial call and return symbols to mark token boundaries. This process is formalized next.
Given a tokenizer \( \tk \) with \( H_{\tk}=H_{\text{call}}\cup H_{\text{plain}}\cup H_{\text{ret}} \), we first build an extended character set \( \tilde{\Sigma}_{\tk} \):
for the \( i \)-th pair of call and return tokens \( h_{a_i} \) and \( h_{b_i} \), we generate a pair of call and return symbol \( \ac_i \) and \( \retsym{b_i} \) outside of \( \Sigma \). 
 We define \( \tilde{\Sigma}_{\tk} \) as \( \Sigma\cup\{ \ac_i\mid i\in[1..|H_{\text{call}}|] \} \cup\{\retsym{b_i} \mid i\in[1..|H_{\text{ret}}|]\} \).
Then, the transformation of \( s\in\lang \) into language \(
\tilde{\lang} \) over \( \tilde{\Sigma}_{\tk} \) proceeds as follows.
Let \( \tk(s)=\tk(s_1\dots s_{k})=t_1\dots t_{k} \), where \( s_k\in
t_k \). We construct \( \tilde{s} \) based on tokenization:
for each \( i\in[1..k] \), if \( t_i \) belongs to \( H_{a_j}
\), the call symbol \( \ac_j \) is added before \( s_i \) in \( s
\); if \( t_i \) is in \( H_{\text{ret}} \), the return symbol
\( \retsym{b_j} \) is added after \( s_i \) in \( s \). 
For instance, for the XML grammar in Figure~\ref{fig:example_xml}, with the call-return token pair being (\texttt{OPEN},\texttt{CLOSE}), our extended character set \( \tilde{\Sigma}_{\tk} \) have two additional characters, say \( \arta \) and \( \artb \). The seed string \( s=\texttt{<p><p>p</p></p>} \) is converted
to \( \arta \texttt{<p>} \arta \texttt{<p>} p \texttt{</p>} \artb \texttt{</p>} \artb \). 
Note that the resulting string after conversion is a well-matched string
in a character-based VPL that has the call symbol \( \arta \) and return symbol \( \artb \).
This allows us to reuse the previous algorithm for learning a character-based VPA.

With the converter function defined, we can model an equivalence query.
\( \mathcal{E}(\vpa_\tk,\tk) \) takes a hypothesis VPA \( \vpa_\tk \) and a hypothesis tokenizer \( \tk \)
and returns none when the oracle
language is equivalent to the unconverted language recognized by
\( \vpa_\tk \) and otherwise returns some \( s \) such that
\( \memquery_\lang(s) \neq \memquery_{(\vpa_\tk,\tk)}(s) \), where 
\[
\memquery_{(\vpa_\tk,\tk)}(s) =
\begin{cases} 
\text{True} & \text{if } \convert_\tk(s) \text{ is accepted by } \vpa_\tk, \\
\text{False} & \text{otherwise}.
\end{cases}
\]

A learner achieves \emph{exact learning} if $\forall s \in \Sigma^*,\; \memquery_\lang(s) =
\memquery_{(\vpa_\tk,\tk)}(s)$.

Similar to the method we discussed in Section~\ref{subsec:learn_tag}, we utilize nesting patterns  to identify the boundaries of call and return tokens. Our objective is to discover a \emph{compatible} tokenizer, which ensures that the language \( \tilde{\lang}_{\tk}=\{ \convert_\tk(s) \mid s\in\lang \} \) is a VPL. We will demonstrate the existence of a finite set of substrings from which our algorithm can successfully learn a compatible tokenizer. Then the aforementioned converter function transforms the oracle language into a character-based VPL, which according to Theorem~\ref{thm:learn_vpa} 
can be efficiently learned by Algorithm~\ref{alg:learn_vpa}.

\edited{
\paragraph{\edited{Assumptions for oracle languages and tokenizers}}
We previously defined a tokenizer as a function that maps a string to a list of tokens. However, assuming an arbitrary tokenizer is insufficient, as it has been demonstrated that any CFG can be mapped to a VPG through some tagging~\cite{10.1145/1007352.1007390}. 
Take, for instance, the CFG
\( \{ L \to c L c \mid c \} \). A tokenizer might tag \( c \) differently
based on its position within a string, e.g., maps the string
\( ccc \) to the token list \( [\callsym{c}, c, \retsym{c}] \), where
\( \callsym{c} \) and \( \retsym{c} \) represent call and return
tokens, respectively; the resulting language is a VPG. 
To simplify tag learning to the setting where tagging is context independent,
the oracle tokenizer and the language are assumed to satisfy the following properties:


\textbf{Tokenization Consistency}.
  For a string \( p = p_1 \ldots p_k \), if each substring \( p_i \) belongs to a token \( t_i \), 
  \edited{then \( \tk(p) = t_1 \ldots t_k \)}. For example, string \texttt{<p></p>} can be split into \texttt{[<p>,</p>]}; this assumption requires it to be tokenized as \texttt{[OPEN, CLOSE]}.

\textbf{Separation}.
  Strings for different tokens do not overlap.
  For the previous example of \( \{ L \to c L c \mid c \} \), tagging the first \( c \) as a call token and the last as a return token would violate this property.

\textbf{Exclusivity}.
  A prefix or suffix of a call or return token \(h\) cannot serve as an infix of \(h\).
  Exclusivity is not required for a token that contains only a single character.

\textbf{Unique Pairing}.
  Each call token is uniquely paired with a return token, similar to an assumption
  we made for the setting of character-based VPLs.

\textbf{Token Fixed Prefix and Suffix}. 
  For each call or return token \(h\), if \(h\) contains more than a single character,
  we require that there exists a prefix \(q\) and a suffix \(g\), such that all strings of \(h\) starts with \( q \) and ends with \( g \). Further, there exists a string \(s_h\) of \(h\), such that the combination of the prefixes and suffixes of \(s_h\) constitutes a sufficient set of test strings for exact learning of the token using \lstar, from the membership query function \(\lambda s. \memquery_\lang(wsw')\), where \(w,w'\) are any strings such that \( ws_hw'\in \lang \).

\textbf{\( k \)-Repetition}.
  Given a positive numbr \(k\),
  for each valid string \(s=w_1ww_3\) where \(w\) is a nonempty substring,
  we say that \(w\) is \emph{\(k\)-repeatable} in \(s\),
  if \(w_1w^kw_3\) is also a valid string.
  A language \(\lang\) and its tokenizer \(\tk\) are said to satisfy \emph{\(k\)-Repetition} if, for any valid string \(s\in\lang\) and any substring \(w\) in \(s\), if \(w\) belongs to a call or return token \(h\), but is not tokenized as \(h\) in \(s\), then \(w\) is \(k\)-repeatable in \(s\).

  For example, consider the JSON string \texttt{\{"\{":true\}}. Suppose \(w\) is the second \texttt{\{}; since it is inside quotes, it belongs to part of the token for a JSON key (i.e., \texttt{"\{"}), even though $\{$ itself is a call token. For any \(k\), \(w\) is \(k\)-repeatable since the string \(\texttt{\{"\(w^k\)":true\}}=\texttt{\{"\{}\ldots\texttt{\{":true\}}\) remains a valid JSON string.
  In our implementation, we set \(k\) to \(2\).}

\edited{While we make the aforementioned assumptions, our approach is still quite expressive, since the above properties are typically satisfied by practical grammars, including those in our evaluation.}

\subsection{Tagging Inference for Tokens}\label{subsec:infer_token}

In this section, we define the compatibility of tokenizers and present a theorem about the relation between a compatible tokenizer and the converted language \( \tilde{\lang}_{\tk} \). Then, we discuss an algorithm that infers a compatible tokenizer from given seed strings.

To define compatible tokenizers, we introduce some additional definitions. Given a tokenizer \( \tk \), recall that for a string \( s \) in \( \lang \), \( \convert_\tk(s) \) is built by inserting artificial call and return symbols to \( s \). Now, let \( s=s_1s_2s_3 \). We define \( \convert_{\tk,s}(s_2) \) as the substring in \( \convert_\tk(s) \) that corresponds to \( s_2 \); additionally, if {\( \convert_\tk \) inserted} a call symbol between \( s_1 \) and \( s_2 \), 
then \( \convert_{\tk,s}{(s_2)} \) includes and starts with that call symbol; if {\( \convert_\tk \) inserted} a return symbol between \( s_2 \) and \( s_3 \), then \( \convert_{\tk,s}{(s_2)} \) includes and ends with that return symbol.
For example, for the seed string \( s=\texttt{<p><p>p</p></p>} \), \( \convert_{\tk,s}{( \texttt{<p>})}=\arta\ \texttt{<p>} \), and \( \convert_{\tk,s}{( \texttt{</p>})}=\texttt{</p>}\ \artb\).

\begin{definition}[Compatible Tokenizers]\label{def:compatible_tokenizer}
  We say that a tokenizer \( \tk \) is compatible with a set of nesting patterns \( N \), if for each nesting pattern \( s=uxzyv \) in \( N \), \( \convert_{\tk,s}{(x)} \), and \( \convert_{\tk,s}{(y)} \), there exists a pair of artificial call and return symbols \( (\arta,\artb) \) in \( \tk \), such that (1) \( \convert_{\tk,s}{(x)} \) includes \( \arta \) and \( \convert_{\tk,s}{(y)} \) includes \( \artb \), and (2) \( \arta \) is unmatched in \( \convert_{\tk,s}{(x)} \) 
and \( \artb \) is unmatched in \( \convert_{\tk,s}{(y)} \).

  We say that a tokenizer \( \tk \) is compatible with a set of seed strings \( S \), if (1) for each string \( s \) in \( S \), \( \convert_\tk{(s)} \) is well-matched, and (2) \( \tk \) is compatible with all nesting patterns of \( S \).
\end{definition}

Now we present Theorem~\ref{thm:token_vpl} as the basis for exact learning.

\begin{theorem}\label{thm:token_vpl}
  \edited{Assume the oracle language and the oracle tokenizer satisfy the Tokenization Consistency and Separation properties.}
  Given a tokenizer \( \tk \) that is compatible with the oracle language \( \lang \), language \( \tilde{\lang}_{\tk} \) is a VPL.
\end{theorem}

Now, we propose
Algorithm~\ref{alg:backtracking_token} to infer a hypothesis compatible 
tokenizer.  
\edited{Instead of finding a full-fledged tokenizer,
the algorithm infers a \emph{partial tokenizer}, which recognizes
only call and return tokens in an input string; the syntax of plain tokens is instead learned during the VPA learning process.
As a result of this choice, substrings between call/return tokens recognized by a partial tokenizer
are implicitly treated as plain tokens.
We represent a partial tokenizer as a set \( D=\{ (r_i, r_i') \mid i \in [1..|H_{\text{call}}]| \} \), where \(r_i\) and \(r_i'\) are the regular expressions for the \( i \)-th paired call and return token, respectively.
Function $\convert_D{(s)}$ and compatibility are similarly defined for a partial tokenizer \(D\); we omit them for brevity.}

At a high level, Algorithm~\ref{alg:backtracking_token} identifies call/return tokens by enumerating potential prefixes and suffixes based on the Token Fixed Prefix and Suffix assumption. \edited{Further, under the Exclusivity assumption}, \edited{we can prove that oracle call/return tokens must appear in \( (x^2,y^2) \)} for certain nesting pattern \( (x,y) \); the proof is provided in the appendix as Lemma~\ref{app:lemma:token_in_pattern2}. \edited{Therefore, we restrict our enumeration to substrings within \( (x^2,y^2) \). Our approach begins by searching within \( (x,y) \) and then progressively expands the search space to \( (x^2,y^2) \).} 
Upon identifying a candidate prefix-suffix pair for a token, we \edited{learn the token's lexical rules as a regular expression within the prefix-suffix pair using \lstar at line 6; in this learning, we simulate the equivalence queries using test strings obtained by combining the prefixes and suffixes of \(x\) and \(y\), respectively.} We then incorporate the tokens into the partial tokenizer and proceed to assess the tokenizer's compatibility with the nesting patterns of seed strings (line 7).

One compatibility condition is that the seed strings  after tokenization should be well-matched; for that, 
Algorithm~\ref{alg:tokenize} is used to tokenize a string based on a given partial tokenizer.
\edited{The main challenge of tokenization is that we have only a partial tokenizer and we need to rely on 
$k$-Repetition to deal with the
case when a plain token string contains a call/return token as part of its substring. E.g., in 
\(s=\texttt{\{"\{":true\}}\), the second \texttt{\{} is actually part of the plain token \texttt{"\{"} and should \emph{not} be 
treated as a call token.}
\edited{To demonstrate Algorithm~\ref{alg:tokenize}, consider a partial tokenizer \(D=\{ (\texttt{\{}, \texttt{\}}) \}\) and an input JSON string \(s=\texttt{\{"\{":true\}}\). Start with string index \(i=1\) and token list \(l=[]\), Algorithm~\ref{alg:tokenize} matches the first \texttt{\{} as \(m_1\) and pushes it to \(l\). Since \(i=2\) does not result in any match, \(i\) is updated to \(3\), 
where the second \texttt{\{} is matched as \(m_2\); however it is not added to the token list since it is \(k\)-repeatable. Finally, the last \(\texttt{\}}\) is matched as \(m_3\). As a result, Algorithm~\ref{alg:tokenize} returns \([m_1,m_3]\).}

\begin{algorithm}[t]
  \SetAlgoLined
  \caption{The \( \tokenInfer{(\oracle, S)} \) algorithm that infers call and return tokens. The \( \candidateNP{-}{-} \) function is the same as the one in Algorithm~\ref{alg:backtracking_tagging}.}
  \label{alg:backtracking_token}

  \KwIn{Oracle \( \oracle \) and seed strings \( S \).}
  \KwOut{Some tokenizer \( D \) compatible with \( S \), or \texttt{None} if no compatible tokenizer is found.}

\SetKwProg{Fn}{Function}{:}{}

\SetKwFunction{TokenSearch}{\( \tokenSearch \)}
\Fn{\TokenSearch{\( N \), \( \ndone \), \( D \)}}{
  \lIf{\( N \) is empty}{\Return{\( \texttt{Some}(D) \)}}
  Take a nesting pattern \( s=uxzyv \) from \( N \)\;
  \If{\( D \) is incompatible with \( uxzyv \)}{
    \ForEach{disjoint substrings \( q \) and \(g\) in \edited{\(x\)} and \( x^2 \), and \( q' \) and \( g' \) in \edited{\(y\)} and \( y^2 \)
    }{
      \edited{Based on \(((q, g), (q', g'))\) learn a new call-return token pair \((r,r')\)}\;
      \If{\( D\cup\{(r,r')\} \) is compatible with \( \ndone\cup\set{uxzyv} \)}{
        \( D' \gets \tokenSearch{(N \setminus \set{uxzyv}, \ndone\cup\set{uxzyv}, D\cup\{(r,r')\} )} \)\;
        \lIf{\( D' \) is not \texttt{None}}{\Return{\( D' \)}}
      }
    }
    \Return{\texttt{None}}; \tcp{No compatible tokenizer found}
  }
  \lElse {\Return{\( \tokenSearch{(N \setminus \set{uxzyv}, \ndone\cup\set{uxzyv}, D)} \)}}
}
Initialize \( K \) as \( 1 \)\;
\Repeat{\( D \neq \texttt{None} \)}{
  \( K\gets K+1 \); \( N_{S,K} \gets  \candidateNP{S}{K} \);
  \( D\gets \tokenSearch{(N_{S,K},\emptyset,\emptyset)} \);
}
{\Return{\( D \)}\;}

\end{algorithm}

\begin{algorithm}[t]
  \SetAlgoLined
  \caption{The \( \tokenit(D,s) \) algorithm that tokenizes a string.}
  \label{alg:tokenize}

  \KwIn{A partial tokenizer \( D \) and a string \( s \).}
  \KwOut{A token list \( l \).}

  Initialize the token list as \( l\gets[] \)\;
  Initialize the current location of string \( s \) as \( i\gets 1 \)\;
  \While{\( i\leq|s| \)} {
    \If{ We find a first match \( w=s[i]...s[j] \) for \edited{token \( h \in D \)} and \(w\) is not \(k\)-repeatable}{
      Push new match \( m=(h,i,j) \) to token list \( l \)\;
      \( i\gets j+1 \)\;
      }
      \lElse{\(i\gets i+1\)}
    }
  \Return{\( l \)}\;
\end{algorithm}

We next illustrate the steps of Algorithm~\ref{alg:backtracking_token} using our XML example.
We start with an empty tokenizer \( D \). We then iteratively select a nesting pattern \( s = uxzyv \), tokenize \( s \) using Algorithm~\ref{alg:tokenize}, and verify the compatibility of \( D \) with the tokenization. 
In our example, we start with the seed string \( s=\texttt{<p><p>p</p></p>} \)
and pick a nesting pattern. Suppose \vstar picks the outmost pattern \( (\texttt{<p>},\texttt{</p>}) \).
The token list \(\tokenit(D,s)\) of \( s \) is empty, since there is no rule to find a token yet.
Apparently, this tokenizer \( D \) is not compatible with \( uxzyv \).

We then extend \( D \) by a call-return token pair learned from \( ((q, g), (q', g')) \) derived from \edited{\((x,y)\) or} \( \left(x^2, y^2\right) \).
By enumerating candidate prefixes and suffixes \( ((q, g), (q', g')) \) within \( (x,y)=(\texttt{<p>},\texttt{</p>})\),
we can build the first call-return token pair.
From \( x \), we first pick the outmost \edited{\((\texttt{<}, \texttt{>})\)} as \( (q,g) \); then in \( y \), 
we pick the outmost \edited{\((\texttt{<}, \texttt{>})\)} as \( (q',g') \).
\edited{By learning the tokens' lexical rules from membership query functions \(\lambda w. \memquery_\lang(w\texttt{<p>p</p></p>})\) and \(\lambda w. \memquery_\lang(\texttt{<p><p>p</p>}w)\), we identify two regular expressions \texttt{<p>} and \texttt{</p>} for the call and return tokens, respectively.
Note that if the open tag contained XML attributes, the learned lexical rules would encompass regular expressions that specify these attributes.}
To check if the partial tokenizer \( D=\{ (\texttt{<p>}\), \(\texttt{</p>}) \} \) is compatible with \( s=\texttt{<p><p>p</p></p>} \), we need to tokenize \( s \) following Algorithm~\ref{alg:tokenize}, \edited{which returns the token list
\([\texttt{<p>},\texttt{<p>},\texttt{</p>},\texttt{</p>}]\).}
It can be shown that this partial tokenizer is compatible with all nesting patterns of string \( s \).
Therefore, Algorithm~\ref{alg:backtracking_token} ends here and returns this compatible tokenizer. 

\begin{lemma}[Finite and Sufficient Seed Strings]\label{lemma:finite_seed_strings_for_tokenization}
  \edited{Assume the oracle language and the oracle tokenizer satisfy the Tokenization Consistency, Separation, Exclusivity, Unique Pairing, Token Fixed Prefix and Suffix, and \(k\)-Repetition properties.}
  There exists a finite set of seed strings \( S_0\subseteq \lang \), with which we can find a tokenizer that is compatible with the oracle language \( \lang \) using Algorithm~\ref{alg:backtracking_token}.
\end{lemma}

As a summary, we can learn a compatible tokenizer from a certain
finite set of seed strings. With a compatible tokenizer \( \tk \),
Theorem~\ref{thm:token_vpl} gives us that \( \tilde{\lang}_{\tk} \) is a
character-based VPL. Then by Theorem~\ref{thm:learn_vpa}, we can use
Algorithm~\ref{alg:learn_vpa} to learn \( \tilde{\lang}_{\tk} \) exactly
under active learning.


\section{Evaluation}\label{sec:Evaluation}
In this section, we discuss \vstar's implementation, evaluation and
its comparison with two other state-of-the-art grammar inference
tools, \glade \cite{10.1145/3140587.3062349} and \arvada \cite{10.1109/ASE51524.2021.9678879}, in the context of inferring grammars from
program inputs.


\paragraph{Implementation}
While black-box programs naturally support membership queries,
direct support of equivalence queries is absent.
\edited{To instantiate the MAT,
we approximate equivalence queries through
membership queries. In particular, we construct a set of strings
by combining prefixes, \edited{infixes, }and suffixes of the seed strings; for
each such string \( s \), if \( \convert_\tk(s) \) is well-matched, 
we add it to a set of test strings. 
The set of test strings is then
used to check the consistency between the hypothesis VPA and the
oracle language. A test string becomes a counterexample if it 
witnesses inconsistency (i.e., either the hypothesis VPA or the oracle accepts the string, but not both)}.
\edited{Similar ideas have appeared in conformance testing \cite{10.1007/11817949_14,10.1145/3605360}.}


Our previously discussed algorithm produces a visibly pushdown automaton (VPA),
instead of a visibly pushdown grammar (VPG).  Upon the successful learning
of a VPA, we transform it into a VPG using methods outlined by
\citet{10.1145/1007352.1007390}.

\paragraph{Datasets} For our experiments, we replicated the evaluation methodology of the \arvada study, utilizing their 
datasets~\cite{10.1109/ASE51524.2021.9678879}, including the oracle grammars, datasets for evaluating the recall (discussed later), and seed strings. 
We selected the grammars of JSON, LISP, XML, While, and MathExpr, due to their distinct characteristics of being VPGs. 


\paragraph{Metrics}
We evaluate the performance of \vstar using four key metrics: Recall,
Precision, F-1 Score, and Number of Membership Queries. We define each
metric as follows:

\begin{enumerate}
    \item \textbf{Recall}: This metric is the probability that a string of the oracle grammar is also a string of the learned grammar \( G \) \cite{10.1145/3140587.3062349}. For finite languages, it can be defined as:
    \(
    \frac{|\lang_\oracle \cap \lang_G|}{|\lang_\oracle|}
    \).
    Due to the potential infinity of the languages, it may be impractical to compute recall directly. Instead, we approximate it by using a representative dataset from the oracle language and then calculating the proportion of this dataset that is accepted by the learned grammar.

    \item \textbf{Precision}: Contrary to recall, precision is the probability that a string in the learned language is accepted by the oracle \cite{10.1145/3140587.3062349}. For finite languages, it can be defined as:
    \(
    \frac{|\lang_\oracle \cap \lang_G|}{|\lang_G|}
    \).
    As with recall, we approximate precision by sampling strings from the learned grammar and calculating the percentage of strings that are accepted by the oracle. We adopt the same sampling method from \arvada \cite{10.1109/ASE51524.2021.9678879}.
    
    \item \textbf{F-1 Score}: The F-1 score is the harmonic mean of precision and recall, defined as 
    \(
      \frac{2}{\frac{1}{R}+\frac{1}{P}}
    \).
    where \( R \) is recall and \( P \) is precision. The F-1 score serves as a measure of the overall accuracy, only reaching high values when both precision and recall are high.
    
    \item \textbf{Number of Unique Membership Queries}: This counts the number of unique membership queries, i.e., distinct oracle calls, made during the learning process. Since a particular string might be queried multiple times, we cache the result after the first query, and only count unique queries. This metric serves as an efficiency measure.
\end{enumerate}

\begin{table}[t]
\caption{Evaluation on datasets where the oracle grammars are VPGs.  ``\#Seeds'' is the number of seed strings for each grammar. 
``\#Queries'' is the number of membership queries, while ``\%Q(Token)'' and ``\%Q(VPA)'' are the percentages of these queries attributed to token inference and VPA learning, respectively. ``\#TS'' is the number of test strings sampled by V-Star. \secondEdited{Results for Arvada are listed as the means over 10 runs ± the standard deviation \cite{arvadaArtifact}.}}
\label{tab:benchmark1}
\centering
\begin{tabular}{lrrrrrr}
  \toprule
   & & \multicolumn{5}{c}{\glade} \\
  \cmidrule(lr){3-7}
   & \#Seeds &\multicolumn{1}{c}{Recall} & \multicolumn{1}{c}{Precision} & \multicolumn{1}{c}{F1} & \multicolumn{1}{c}{\#Queries} & \multicolumn{1}{c}{\edited{Time}}  \\
  \midrule
  json & 71 & 0.42 & 0.98 & 0.59 & \SI{11}{K} & \SI{21}{s} \\
  lisp & 26 & 0.23 & 1.00 & 0.38 & \SI{3.8}{K} &\SI{7}{s} \\
  xml & 62 & 0.26 & 1.00 & 0.42 & \SI{15}{K} & \SI{21}{s} \\
  while & 10 & 0.01 & 1.00 & 0.02 & \SI{9.2}{K} & \SI{13}{s} \\
  mathexpr & 40 & 0.18 & 0.98 & 0.31 & \SI{19}{K}& \SI{42}{s} \\
\bottomrule
\end{tabular}

\begin{tabular}{lrrrrrr}
  \toprule
  & \multicolumn{5}{c}{\arvada} \\
   \cmidrule(lr){2-6}
   & \multicolumn{1}{c}{Recall} & \multicolumn{1}{c}{Precision} & \multicolumn{1}{c}{F1} & \multicolumn{1}{c}{\#Queries} & \multicolumn{1}{c}{\edited{Time}} \\
 \midrule
 json & 0.97 ± 0.09 & 0.92 ± 0.08 & 0.94 ± 0.05  & \SI{6.8}{K} ± 394 & \SI{25}{s}  ± \SI{2}{s}\\
 lisp & 0.38 ± 0.26 & 0.95 ± 0.08 & 0.50 ± 0.18 & \SI{2.2}{K} ± 307  & \SI{8} {s} ± \SI{2}{s}\\
 xml & 0.99 ± 0.02 & 1.00 ± 0.00 & 1.00 ± 0.01  & \SI{12}{K} ± \SI{1}{K} & \SI{61}{s} ± \SI{5}{s} \\
 while & 0.91 ± 0.20 & 1.00 ± 0.00 & 0.94 ± 0.14  & \SI{5.4}{K} ± 563 & \SI{15}{s} ± \SI{1}{s} \\
 mathexpr & 0.72 ± 0.24 & 0.96 ± 0.03 & 0.80 ± 0.16  & \SI{6.6}{K} ± 421  & \SI{24}{s}  ± \SI{2}{s}\\
  \bottomrule
\end{tabular}

\begin{tabular}{lrrrrrrrrrrrrr}
  \toprule
  & \multicolumn{7}{c}{\vstar} \\
  \cmidrule(lr){2-9}
    & \multicolumn{1}{c}{Recall} & \multicolumn{1}{c}{Precision} & \multicolumn{1}{c}{F1} & \multicolumn{1}{c}{\#Queries} & \multicolumn{1}{c}{\edited{\%Q(Token)}} & \multicolumn{1}{c}{\edited{\%Q(VPA)}} & \multicolumn{1}{c}{\edited{\#TS}} & \multicolumn{1}{c}{\edited{Time}} \\
  \midrule
  json     & 1.00 & 1.00 & 1.00 & \SI{541}{K}   & \SI{2.71}{\%}  &  \SI{97.29}{\%} & 8043 & \SI{33}{min} \\
  lisp     & 1.00 & 1.00 & 1.00 & \SI{16}{K}    & \SI{1.37}{\%}  &  \SI{98.63}{\%} & 693  & \SI{77}{s} \\
  xml      & 1.00 & 1.00 & 1.00 & \SI{208}{K}   & \SI{94.93}{\%} & \SI{5.07}{\%}  & 682  & \SI{16}{min} \\
  while    & 1.00 & 1.00 & 1.00 & \SI{1440}{K} & \SI{9.40}{\%}  &  \SI{90.60}{\%} & 119  & \SI{1.5}{h} \\
  mathexpr & 1.00 & 1.00 & 1.00 & \SI{4738}{K} & \SI{0.11}{\%}  &  \SI{99.89}{\%} & 2602 & \SI{6}{h} \\
  \bottomrule
\end{tabular}
\end{table}

\paragraph{Results}

Table~\ref{tab:benchmark1} summarizes the performances of \glade, \arvada, and \vstar on oracle VPGs, \secondEdited{with the results of \arvada and \glade assessed on the same platform as V-Star, utilizing the \arvada artifact \cite{arvadaArtifact}.} 
The table shows that \vstar achieves exact learning for \edited{all} oracles, exhibiting superior accuracy compared to other tools. 
However, \vstar issues a greater number of queries than \glade and \arvada, \edited{resulting in greater inference time.} 
This primarily stems from (1) the substantial number of {test strings} used in approximating equivalence queries,
and (2) the fact that \vstar consumes seed strings without pre-processing. In contrast, \glade and \arvada employ a pre-tokenization strategy, such as grouping digits or letters as a single terminal, which reduces seed string lengths. We take our approach since \vstar can learn tokens. 
\edited{Overall the evaluation shows that \vstar is more accurate but takes more time to infer grammars.
In grammar learning, we believe that accuracy is a more important goal as a more accurate grammar benefits downstream applications greatly. Improving efficiency of \vstar (e.g., using heuristics of target grammars) while not decreasing accuracy is left for future work.
}

\vstar requires a considerable number of membership queries for the MathExpr grammar. This can be attributed in part to the large number of constant function names (26 in all) within the grammar, such as ``sin'' or ``cos''. In its quest for high accuracy, \vstar explores various combinations of these constant names exhaustively. We acknowledge that this approach could be further optimized
and propose this as an avenue for future improvement.

\edited{In Table~\ref{tab:benchmark1}, we include data on the
  percentage of membership queries allocated for token inference
  (``\%Q(Token)'') and for learning VPA (``\%Q(VPA)'').} It can be
observed that the majority of queries are utilized for VPA
learning. This is mainly because seed strings tend to be short,
leading to fewer potential nesting patterns. \edited{One exception is
  XML, where most queries are for token inference. This is because 
  the XML grammar, primarily based on nested tag pairs,
  allows for easier inference of the overall grammar once the
  opening and closing tags (call and return tokens) are
  identified. Furthermore, many queries are required to infer the lexical
  rules of XML attributes.} Additionally, the table provides
information on the count of seed strings (``\#Seeds'') used in our
evaluation. For the grammars assessed, V-Star requires a relatively
small number of seed strings to achieve exact learning, attributed to
\edited{its strategy of employing a wide range of substring combinations to
construct test strings for effective simulating equivalence queries; column ``\#TS'' shows
the number of test strings constructed}.

\section{Future Work}

We believe the performance of \vstar can be further improved with more
advanced methods for generating counterexamples, such as using
machine learning tools to infer counterexamples from seed strings
and a VPA.  A related direction is to investigate the potential
adaptation of \vstar with discrimination trees. Other grammar
inference tools that are based on discrimination trees such as
\ttt~\cite{Isberner2015FoundationsOA} enhance inference efficiency by
reducing counterexample lengths and minimizing membership queries. It
remains to be seen how \vstar can be adapted in this manner and what
improvements this can yield.

The present study focuses primarily on inferring well-matched VPGs
using \vstar. However, our preliminary experience suggests that \vstar
can also be effectively employed to learn general
VPGs~\cite{Alur2005Congruences} with open call and return symbols.  A
general VPG can be used to specify streaming data. As such, the
learning problem for general VPGs is a promising direction for future work.

\edited{
    Our method makes the assumption of unique call-return token pairing to ease tokenizer inference and reduce computational complexity, as matching one call token with multiple return tokens complicates tokenizer inference. It would be interesting future work to consider the implications of relaxing this assumption to enhance flexibility.}

\edited{Experimentally, we focus on languages such as XML and JSON to align with benchmarks used by prior tools for a direct comparison. It would be interesting to evaluate \vstar on more complex programming language grammars to check its effectiveness on those grammars.}

\edited{Improving the readability of VPGs inferred by V-Star is still a challenge. Currently, the grammars generated tend to be larger and less readable than oracle grammars, due to the inherent rigid requirements of VPG rules, the inclusion of lexical rules, and automatically named nonterminals. Although we have made attempts to refactor grammars using regular expressions, these solutions are largely heuristic and may not consistently yield optimal results. Exploring machine learning-based approaches presents a promising avenue to systematically enhance the clarity and conciseness of inferred grammars, making them potentially more accessible and understandable for users.}

Finally, VPGs learned by \vstar may provide a valuable starting point
for better inference algorithms of CFGs. For instance, similar to the
CFGs learned by \glade~\cite{10.1145/3140587.3062349}, the VPGs inferred by
\vstar can serve as inputs for machine learning tools such as
\reinam~\cite{10.1145/3338906.3338958}, which improves the input grammar with
reinforcement learning. Comparing the improvements enabled by these
different starting grammars would be an intriguing line of inquiry.


\section{Conclusions}\label{sec:Conclusion}
This paper introduces \vstar, an algorithm designed to take advantage
of nesting structures in languages to achieve exact learning of
visibly pushdown grammars. Through a set of novel techniques to infer
token boundaries and tag call/return tokens, \vstar demonstrates
its capability to learn a diverse array of practical languages. Our
preliminary experiments are promising and show \vstar's advantages
of accurate learning.

\newpage

\appendix

\section{Proofs of Theorems in Section~\ref{subsec:learn_vpa}}\label{app:subsec:learn_vpa}

\begin{theorem}\label{app:thm:learn_vpa}
  {Given a tagging function \( t \) such that language \( \hat\lang=\{ t(s) \mid s\in\lang \} \) is a VPL,} the minimal
  \( k \)-SEVPA of language \( \hat\lang \) can be learned in polynomial numbers of equivalence and membership queries.
\end{theorem}

\begin{proof}
  We run Algorithm~\ref{alg:learn_vpa} with a target language \(\hat \lang\) over the alphabet \(\hat\Sigma\). Define \(m\) as the state count of the minimal \ksevpa of \(\hat \lang\) and \(n\) as the maximum length
of counterexamples returned by equivalence queries.
  Proposition~\ref{prop:bound_states} establishes that the number of equivalence queries does not exceed \(m\), as each iteration expands the number of states in \( \qc \) by a minimum of one. This also shows that the algorithm must terminate.

  A counterexample returned by an equivalence query causes at most \(\log n\) membership queries
  as detailed in Proposition~\ref{prop:find_counterexample},  resulting
  in no more than \(m \log n\) membership queries during Step~4 of
  Algorithm~\ref{alg:learn_vpa}.
  Membership queries in Steps~2 and~5 involve words either of form
  \(wqw'\) or \(wqmw'\), where \(q \in Q_i\), \(m \in \Sigma_M\), and
  \( (w,w') \in C_i\). With \(|C_i|\) bounded by \(|Q_i| \leq m\) at
  completion, total queries amount to at most \(\sum_{i=0}^{k}(|Q_i| +
  |Q_i||\Sigma_M|)|C_i| = \sum_{i=0}^{k}|Q_i||C_i|(1 + |\Sigma_M|) \leq m^2(1 + |\Sigma| + |\callsyms|\times m
  \times|\retsyms|) \leq m^3|\Sigma|^2 \).

  In conclusion, the number of queries remains polynomially bound by \(n\), \(m\), and \(|\Sigma|\), {including \( O(m^3|\Sigma|^2+m\log n) \) membership queries and \( O(m) \) equivalence queries.}
\end{proof}

\section{Proofs of Theorems in Section~\ref{subsec:learn_tag}}\label{app:proofs_of_learn_tag}

Given tagging \( t \), we say that string \( s \) is \emph{\( t \)-well-matched}, if \( t(s) \) is well-matched.

\begin{definition}\label{def:parse_tree}
  [{Parse Tree}]
  Given a grammar \( G = (V, \Sigma, P, L_0) \), a parse tree with respect to grammar \( G \) is an ordered tree where (1) the leaves of the tree are terminals in \( \Sigma \) or \( \epsilon \), and (2) each non-leaf node is a nonterminal \( L \) in \( V \), where the children of the node are \( \alpha_1 \), \( \alpha_n \), \dots, \( \alpha_n \) such that \( L\to \alpha_1 \alpha_2 \ldots \alpha_n \) is a production rule in \( P \), or \( \epsilon \), such that \( L\to\epsilon \) is a production rule. The root of the tree should be \( L_0 \), the start nonterminal of grammar \( G \).  A parse tree of a string \( s\in\Sigma^* \) in the language of grammar \( G \) is a parse tree whose leaves, when concatenated from left to right, form \( s \).
\end{definition}

\begin{lemma}
  [Pumping Lemma for VPLs]\label{app:lemma:pumping}
  For any VPL \( \hat\lang \), there exists a positive number \( l \) such that, for any string \( s \) in \( \hat\lang \) with length greater than \( l \), it is possible to express \( s \) according to one of the following conditions:
  \begin{enumerate}
  \item (Regular Pumping) We can partition \( s \) into \( s = uxv \) for strings \( u, x, \) and \( v \), with \( x \) being non-empty, such that \( ux^kv \) remains in \( \hat\lang \) for all \( k \geq 0 \). 
  \item (Nesting Pumping) We can partition \( s \) into \( s = uxzyv \) for strings \( u, x, z, y, \) and \( v \), with \( x \) and \( y \) being non-empty, \( x \) containing a call symbol, and \( y \) containing a return symbol, such that \( ux^kzy^kv \) is valid for all \( k \geq 1 \).
  \end{enumerate}
\end{lemma}
  
\begin{proof}
  Let VPG \( G=(\hat\Sigma, V, P, L_0) \) be a grammar of \( \hat \lang \).
  We define \( l \) as the length of the longest string \( s \) that contains no recursion in {any of} its parse trees. Formally, in a parse tree of \( s \), if the subtree rooted with a nonterminal node \( A \) contains another appearance of \( A \), we say there is recursion in the parse tree and we call \( A \) a recursive nonterminal in the parse tree.
\( l \) is then defined as the length of the longest string \( s \) whose parse trees do not contain recursion.
This \( l \) is well defined because the number of non-recursive parse trees is finite: any path that goes from the root to a leaf of a parse tree and exceeds
the length of \( |V|+2 \) must have \( |V|+1 \) nonterminals and revisit at least one nonterminal twice.

  For any string \( s \) exceeding \( l \) in length, one of its parse trees must have a recursive nonterminal; say it is \( L \). The derivation of the parse tree can be written as:
  \( L_0 \to^* u L v \to^* u (s_1Ls_2)v \to^* s \),  where \( L \to^* s_1Ls_2 \). We have two cases:

  \begin{enumerate}
    \item If \( s_2 \) is empty, then \( s_1 \) cannot be empty since \( L \to L \) is not a valid VPG rule. Thus, \( u(s_1)^ks_Lv \) remains valid, where \( s_L \) is {a} terminal string derived from \( L \). This satisfies regular pumping. 

    \item If \( s_2 \) is not empty, then a matching rule is used somewhere in the derivation sequence that leads to the second appearance of \( L \). This is because
     by the VPG rules, if only rules of the form \( L_1 \to c L_2 \) or \( L_1 \to \epsilon \) were used, then \( L \) must be the last symbol in the derived string \( s_1Ls_2 \), which contradicts with that \( s_2 \) is not empty.   This leads to:
    \[ L \to^* s_1'\ac A \bc B s_2' \to^* s_1'\ac s_1''Ls_2'' \bc B s_2' \]
    Here, \( A \to^* s_1'' L s_2'' \). We then select \( x \) as \( s_1'\ac s_1'' \) and \( y \) as \( s_2''\bc s_2' \) for nesting pumping.
  \end{enumerate}
\end{proof}

\begin{lemma}\label{app:lemma:verify_nesting_pattern_for_string}
  Consider an oracle VPL \( \lang \), a VPG \( G=(\Sigma, V, P, L_0) \) for \(L\), and an oracle tagging \( t_\oracle \). For each string \( s \in \lang \) and \( s=uxzyv \), where \( u,x,z,y,v \) are substrings, and \( x,y \) are nonempty, if string \( u  x^k  z  y^k  v \) is valid but string \( u  x^k  z  y^j  v \) is invalid for \( k,j\leq (|V|^2+1)^2\) and \(k\neq j\), then \( t_\oracle(x) \) contains an oracle call symbol, and \( t_\oracle(y) \) contains an oracle return symbol, and the two symbols are matched with each other in \(s\).
\end{lemma}
\begin{proof}
  In this proof, we abuse the notation \( x \) to also mean \( t_\oracle(x) \).
  
  We first show that \(x\) as well as \(y\) contains unmatched symbol. Otherwise, \( x \) and \( y \) contain only plain symbols or well-matched call-return pairs. For each \( k \geq 1 \) and string \( u  x^{k}  z  y^{k}  v \), the derivation path of \( {x}^{k} \) can be written as \( L_{k,1}\to^* {x}^{k}L_{k,2} \),
  where \( L_{k,1}, L_{k,2} \) are two nonterminals in \( V \).
  This is because, suppose \( x \) contains only plain symbols, then the derivation path of \( x \) is of the form 
  \[ L_1\to {x}[1]L_2\to {x}[1]{x}[2]L_3\to \cdots \to xL_{|x|+1} \]
  for certain nonterminals \( L_{i,i\in[1..|x|+1]} \). The case is similar when \( x \) also contains well-matched substrings; we omit the discussion for brevity.
  Now, for each \( k\in[1..|V|^2+1] \), we have the following derivations:
  \begin{align*}
    L_{1,1}&\to^* {x}^{1}L_{1,2} \\
    L_{2,1}&\to^* {x}^{2}L_{2,2} \\
    &\dots\\
    L_{|V|^2+1,1}&\to^* {x}^{|V|^2+1}L_{|V|^2+1,2}
  \end{align*}
  Apparantly, there exist \( k' \) and \( k_1\neq k_2 \), such that \( k',k_1,k_2\leq |V|^2+1 \), and a pair \( (L_{k',1}, L_{k',2}) \) appears twice on both sides, i.e.,
  \begin{align*}
    &\dots\\
    L_{k',1}&\to^* {x}^{k_1}L_{k',2} \\
    &\dots\\
    L_{k',1}&\to^* {x}^{k_2}L_{k',2} \\
    &\dots
  \end{align*}
  Thus, both string \( u  x^{k_1}  z  y^{k_1}  v \) and \( u  x^{k_2}  z  y^{k_1}  v \) are valid. Given that \( k_1\neq k_2 \), this is a contradiction.

  Therefore, \( x \) and \( y \) must include unmatched symbols. Consider the type of the unmatched symbol in \(x\).
   If \( x \) includes a return symbol \( \bc \), where the matched \( \ac \) is {before} \( x \), then \( ux^2zy^2v \) is invalid, because \( u \) has no additional call symbol to match \( \bc \). Thus, \( x \) includes a symbol \( \ac \), whose matched symbol \( \bc \) is {after} \( x \). If \( \bc \) is in \( y \), we are done. Otherwise, \( \bc \) is either in \( z \) or in \( v \). Consider \( ux^2zy^2v \). Since the string is valid, the unmatched \( \ac \) in \(x^2\) must match a return symbol in \( y^2 \). 
   
   In conclusion, \( x \) contains an oracle call symbol, which matches a return symbol in \( y \).
\end{proof}

\begin{theorem}[Termination and Correctness of Algorithm~\ref{alg:backtracking_tagging}]
  \label{app:lemma:backtrack_result}
  Let \( m \) be the number of states of the minimal \( k \)-SEVPA for the oracle VPL.
  There exists a number \( K\leq ((m^2+2m)^2+1)^2 \), with which function \( \tagInfer{\oracle}{S} \) returns a tagging that is compatible with a finite set of seed strings \( S \).
\end{theorem}
\begin{proof}
  First, we show that at least the oracle tagging can be found, which must be compatible with the pattern. This is because, from Lemma~\ref{app:lemma:verify_nesting_pattern_for_string}, when \( K> (|V|^2+1)^2 \), each candidate nesting pattern \( (x,y) \) must either be invalidated, or contain oracle call-return pair unmatched in \( x \) and \( y \), respectively.

  Therefore, since any VPG for the oracle VPL can be used for the checking, we pick the specific VPG {converted from the minimal \( k \)-SEVPA by the method discussed in \citet{10.1145/1007352.1007390}, Theorem 5.3 (Visibly pushdown grammars), where \( |V| \) is bounded by \( m^2+2m \).}
\end{proof}

For Lemmas~\ref{app:lemma:finite_oracle_unmatched_symbols} and~\ref{app:thm:compatible_tagging},
we first introduce another congruence relation by \citet{Alur2005Congruences}. Two \emph{well-matched} strings, \(s_1\) and \(s_2\), are deemed congruent, denoted as \(s_1 \sim s_2\), if their contexts coincide. Specifically,
\[ \forall u,v \in \Sigma^*, us_1v \in \lang \iff us_2v \in \lang. \]
This congruence is an equivalence relation, and \(\lang\) {is} a VPL on \( \hat\Sigma \) if and only if the congruence relation admits a finite number of equivalence classes. 

Given a tagging \( t \), denote the congruence relation over \( \hat\lang_t \) as \( \sim_t \).

Given a compatible tagging function \(t\) and a string \(s\), the following Lemma~\ref{app:lemma:finite_oracle_unmatched_symbols} shows that, if \(t(s)\) is well-matched, then \(t_\oracle(s)\) has a bounded number of unmatched symbols.

\begin{lemma}\label{app:lemma:finite_oracle_unmatched_symbols}
  Given oracle language \( \lang \) and oracle tagging \( t_\oracle \),
  for each compatible tagging \( t \),
  there exists an upper bound positive number, denoted as \( N_t \),
  such that for each string \( s\in\Sigma^* \), if \( s \) is \( t \)-well-matched and there exists context strings \( (w,w') \) such that \( wsw'\in\lang \), then \( t_\oracle(s) \)  contains at most \( N_t \) unmatched oracle call and return symbols.
\end{lemma}
\begin{proof}
  In this proof, to simplify the notation, we use ``\(s\)'' or ``\(w\)'' to also mean strings tagged by the oracle tagging function \(t_\oracle\). For strings tagged by a compatible tagging function \(t\), we use ``\(t(s)\)'' and ``\(t(w)\)'' explicitly.

  To simplify the problem, let us assume there is an oracle VPG that includes only one matching rule; we denote the matching rule as \(L\to \ab{A}B\). 
  As an overview, we show that 
  for string \( s \) that contains \( K \) unmatched oracle call symbols \( \ac \),
  we can construct \( K \) equivalence classes for the oracle congruence relation.
  Therefore, \( N_t \) is bounded by the number of oracle equivalence classes. The case for multiple matching rules can be similarly proved, and we omit it for brevity.
  
  For \(t\)-well-matched string \(s\) with \(wsw'\in\lang\), if \(s\) contains no unmatched oracle symbols, then we are done. Otherwise, assume \(s\) contains no return symbols, and \( K \) unmatched oracle call symbols \( \ac \) (the other cases are similar and we omit them for brevity).
  We can rewrite \( wsw' \) to reflect the derivation of these oracle call and return symbols as follows:
  \[ wsw' = w+q_0(\ac s_3^{(1)}) \dots (\ac s_3^{(K)})s_L^{(1)}+s_L^{(2)}(s_4^{(K)}\bc s_B^{(K)}) \dots (s_4^{(1)}\bc s_B^{(1)})+w'' \]
  where \(w\), \(q_0\), \( s_2^{(i)} \), \( s_3^{(i)} \), \( s_L^{(1)} \), \( s_L^{(2)} \),  \( s_4^{(i)} \), \( s_B^{(i)} \) and \( w'' \) for \(i\in[1..K]\) are strings, and
  \begin{align*}
    s  &= q_0( \ac s_3^{(1)} ) \dots ( \ac s_3^{(K)} ) s_L^{(1)},\\
    w' &= s_L^{(2)}( s_4^{(K)}\bc s_B^{(K)} ) \dots ( s_4^{(1)}\bc s_B^{(1)} )+w''    
  \end{align*}
  and string \( q_0( \ac s_3^{(1)} ) \dots ( \ac s_3^{(K)} )s_L^{(1)}+s_L^{(2)}( s_4^{(K)}\bc s_B^{(K)} ) \dots ( s_4^{(1)}\bc s_B^{(1)} ) \) is derived by:
  \begin{align*}
    L \to \ac A \bc B &\to ( \ac s_3^{(1)} ) L ( s_4^{(1)}\bc s_B^{(1)} )\\
    &\to ( \ac s_3^{(1)} ) \ac A \bc B  ( s_4^{(1)}\bc s_B^{(1)} )\\
    &\to ( \ac s_3^{(1)} ) ( \ac s_3^{(2)} ) L ( s_4^{(2)}\bc s_B^{(2)} ) ( s_4^{(1)}\bc s_B^{(1)} )\\
    &\to ( \ac s_3^{(1)} ) \dots ( \ac s_3^{(K)} ) L ( s_4^{(K)}\bc s_B^{(K)} ) \dots ( s_4^{(1)}\bc s_B^{(1)} )\\
    &\to ( \ac s_3^{(1)} ) \dots ( \ac s_3^{(K)} ) s_L ( s_4^{(K)}\bc s_B^{(K)} ) \dots (s_4^{(1)}\bc s_B^{(1)})
  \end{align*}
  where \( s_L=s_L^{(1)}s_L^{(2)} \).

  Let \( (x_i,y_i)=(\ac s_3^{(i)},s_4^{(i)}\bc s_B^{(i)}) \), \( i\in[1..K] \). Notice that \((x_i,y_i)\) are \( K \) disjoint nesting patterns.
  Since those \((x_i,y_i)\) are exchangeable, we denote each of them as \((x,y)\) when their indices do not matter. With this notation, we can simplify the above formulae as
  \begin{align}\label{app:eq:split1}
    wsw' &= w+q_0x^Ks_L^{(1)}+ s_L^{(2)}y^Kw'', \text{ where} \\
    \label{app:eq:split2}
    s&=q_0x^Ks_L^{(1)} \\
    \label{app:eq:split3}
    w'&=s_L^{(2)}y^Kw''
  \end{align}

  Since \( t \) is compatible with \( \lang \), for \( i\in[1..K] \), by definition, each pattern \( (x_i,y_i) \) contains a call-return pair \( (c_i,d_i) \) of compatible tagging \( t \),
  where \( c_i \) and \( d_i \) are unmatched in \( x_i \) and \( y_i \), respectively.
  Without loss of generality, let us assume these \( (c_i,d_i) \) are the same, denoted as \( (c,d) \).

  Now consider Equation~(\ref{app:eq:split1})-(\ref{app:eq:split3}). For \( i\in[1..K] \), each \( x_i \) contains a symbol \( c_i \), whose matched symbol, denoted as \( d_i' \), is after \( x_i \); similarly, each \( y_i \) contains a symbol \( d_i \), whose matched \( c_i' \) is before \( y_i \).
  Since \( s \) is \(t\)-well-matched,
  for \( i\in[1..K] \), each symbol \( d_i' \) can only locate in \( s \), and each symbol \( c_i' \) cannot locate in \( s \).

  Now, for \( j\in[1..K] \), we construct \( t_\oracle \)-well-matched strings \( s_j \) and their contexts \( (\widehat w_j,\widehat w_j') \), as follows:
  \begin{align*}
    s_j&=x^{K-j}s_L^{(1)} + s_L^{(2)}y^{K-j}\\
    \widehat w_j&=wq_0x^{j}\\
    \widehat w_j'&=y^{j}w''
  \end{align*}
  Now, we prove that each \(s_j\) represents a different equivalence class. First, it is obvious that \(s_j\) is \(t_\oracle\)-well-matched. Then, we show that \(\widehat w_i s_j \widehat w_i'\) is invalid for \(i \neq j\). 
  Let us expand \(\widehat w_i s_j \widehat w_i'\) as
  \[ \widehat w_i s_j \widehat w_i' = w+q_0x^{K+i-j}s_L^{(1)}+s_L^{(2)}y^{K+i-j} w''. \]
  Denote the number of unmatched \(d\) and unmatched \(c\) in string \(x\) as \(n_d(x)\) and \(n_c(x)\), respectively. 
  There are three cases.

  If \(n_d(x) > n_c(x)\), then, when \(i>j\), we have \(n_d(wq_0x^{K+i-j}) > n_d(wq_0x^{K})=0 \), thus \(\widehat w_i s_j \widehat w_i'\)'s prefix contains pending return symbol, therefore the string is invalid.

  If \(n_d(x) < n_c(x)\), then, when \(i<j\), we have
  \begin{equation}\label{app:eq:nc1}
    n_c(\widehat{w}_i)=n_c(wq_0x^i) < n_c(wq_0x^j) = n_c(\widehat{w}_j).
  \end{equation}
  On the other hand, since \(s\) is \(t\)-well-matched, we have 
  \begin{equation}\label{app:eq:nc2}
    n_c(qw_0x^j)=n_d(x^{K-j}s_L^{(1)}).
  \end{equation}
  Therefore, based on Equation~(\ref{app:eq:nc1})-(\ref{app:eq:nc2}), we have
  \[ n_d(\widehat{w}_ix^{K-j}s_L^{(1)}) = n_d(x^{K-j}s_L^{(1)}) - n_c(\widehat{w}_i) = n_c(wq_0x^j) - n_c(\widehat{w}_i) > 0. \]
  Similar to the first case, this shows \(\widehat w_i s_j \widehat w_i'\) is invalid.

  In the last case, \( n_d(x)=n_c(x) \). We show that this is impossible.
  We first assume \( n_d(x)=n_c(x)=1 \), then discuss the other case in the end.
  First, we rewrite each \( x \) in \(x^K\) as
  \[ x = w_1 d w_2 c w_3, \]
  where \( w_{i,i=1,2,3} \) does not contain unmatched \( c \) nor \( d \).
  Therefore, for two adjacent \( x \), we have
  \begin{align*}
    (w_1 d w_2 c w_3)(w_1 d w_2 c w_3)\\
    =w_1 d w_2 (c w_3 w_1 d w_2) c w_3
  \end{align*}
  In string \( c w_3 w_1 d w_2 \), notice that \( c \) is matched with \( d \), therefore, string \( w_3w_1 \) is \( t \)-well-matched. Since \((x,y)\) is a nesting pattern of the oracle tagging function,  for any \(k>0\), we can rewrite \(x^{k+1}\) as
  \[ x^{k+1}= w_1 d w_2 (c w_3 w_1 d w_2)^k c w_3. \] 
  With the corresponding \( y^k \), we have a new nesting pattern 
  \[ (c w_3 w_1 d w_2, y). \]
  Since tagging \( t \) is compatible and \(c w_3 w_1 d\) is \(t\)-well-matched, \( w_2 \) contains an unmatched call symbol of tagging \(t\), denoted as \( g \), whose matched return symbol of tagging \(t\), denoted as \( h \), is after \( w_2 \). 

  Now we have come back to a similar situation of comparing \(n_g(x)\) and \(n_h(x)\). With a similar analysis, we can show that 
  if the number of unmatched \( n_g(x) \neq n_h(x) \),
  then, we can construct \( K \) equivalence classes in \( \sim_{t_\oracle} \).
  Therefore, we again must have \( n_g(x) = n_h(x) \).

  Then, we can rewrite \( x \) by expanding \(w_2\) as
  \[ x = w_1 d w_2 c w_3 = w_1 d (w_1'hw_2'gw_3') c w_3. \]
  And, similarly, rewrite two adjacent \( x \) as
  \begin{align*}
    xx & = (w_1 d (w_1'hw_2'gw_3') c w_3)(w_1 d (w_1'hw_2'gw_3') c w_3)\\
       & = w_1 d w_1'hw_2'(gw_3' c w_3w_1 d w_1'hw_2')gw_3' c w_3
  \end{align*}
  Again, string \( gw_3' c w_3w_1 d w_1'hw_2' \) forms the first part of a nesting pattern, thus must contain another unmatched call symbol in tagging \( t \). 
  
  However, notice that \( gw_3' c w_3w_1 d w_1'h \) is \( t \)-well-matched. Therefore, the new unmatched symbol must appear in \( w_2' \), which is strictly shorter than \(w_2\). Subsequentially, we can find substrings \(w_2'\), \(w_2''\), \dots, \(w_2^{(i)}\), \dots with decreasing lengths that must contain unmatched call symbol in \(t\). However, \(w_2^{|w_2|}\) must be empty and contain no symbol, which makes \( t \) incompatible with \( \lang \), a contradiction.

  Above is the case where the numbers of unmatched \( c \) and \( d \) in \(x\) is \( 1 \). When the numbers are greater than \(1\) (recall that the two numbers should be the same), we rewrite \( x \) as
  \[ w_1 d w_2 c w_3, \]
  where \( d \) is the last unmatched \( d \), and \( c \) is the first unmatched \( c \). Expand \(x^2\) again, and we can observe that string \( w_3 w_1 \) is still \( t \)-well-matched. The rest of the reasoning is the same as above.

  In conclusion, for each \( K \) and \( t \)-well-matched string \( s \) with \( wsw'\in\lang \), the number of unmatched oracle call symbols equals or less than the number of equivalence classes of \( \sim_{t_\oracle} \). Given that \(\lang\) is a VPL under \(t_\oracle\), the numbers of unmatched oracle symbols in any such \(s\) have an upper bound.
\end{proof}

\begin{theorem}\label{app:thm:compatible_tagging}
  Given oracle language \( \lang \) and oracle tagging \(t_\oracle\),
  if language \( \{ t_\oracle(s) \mid s\in\lang \} \) is a VPL, then
  for any tagging \( t \) compatible with \( \lang \), language \( \hat\lang_t=\{ t(s) \mid s\in\lang \} \) is a VPL.
\end{theorem}

\begin{proof}
  By Lemma~\ref{app:lemma:finite_oracle_unmatched_symbols},
  there exists a positive number \( N_t \), such that
  any \( t \)-well-matched string contains at most \( N_t \) unmatched oracle symbols.

  For a given \( t \)-well-matched string \( p \), without loss of generality, let us assume there are only \( K<N_t \) number of unmatched oracle call symbols in \( t \). We prove the theorem by showing that \(t(p)\) is equivalent to a string within a fixed length, denoted as \(N\). 
  
  First, we can partition \(t_\oracle(p)\) into \(t_\oracle(p)=\hat p_1\ac_1\hat p_2\ac_2\dots \hat p_{K}\ac_K\hat p_{K+1}\), where each \(\ac_i\) is an unmatched call symbol and each \(\hat p_i\) is well-matched under \(t_\oracle\). Let us denote the length of the longest representative under \(t_\oracle\) as \(l\). Denote \({[p_i]}_{\oracle}\) as \(p_i\)'s representative in the oracle congruence relation, we have 
  \[ \forall w, w',\ w p_i w'\in \lang \iff w {[p_i]}_{\oracle} w'\in \lang. \]
  Now, we construct a shorter representative for \( t(p) \), by replacing \( p_i \) with \( [p_i]_\oracle \). Formally, for all \(w_1\) and \(w_2\), 
  \begin{align*}
  w_1 p w_2& =w_1(p_1{a}_1p_2{a}_2\dots p_{K}{a}_Kp_{K+1})w_2 \\
  & =(w_1p_1{a}_1\dots {a}_{i-1}) p_i ({a}_i\dots p_{K}{a}_Kp_{K+1}w_2)\in \lang  \\
  & \iff (w_1p_1{a}_1\dots {a}_{i-1}) {[p_i]}_{\oracle} ({a}_i\dots p_{K}{a}_Kp_{K+1}w_2)\in \lang.
  \end{align*}
  Consequently, the length of any representative under tagging \(t\) is limited by \(K+(K+1)l\leq N_t+(N_t+1)l\), accounting for \(K\) unmatched characters and \(K+1\) substrings each with a length not exceeding \(l\). 

  In conslusion, language \( \hat\lang_t \) has a finite number of equivalence classes, therefore is a VPL.
\end{proof}

\begin{theorem}[Finite and Sufficient Seed Strings]
  \label{app:thm:construct_seed_strings}
  For any given oracle language \(\lang\), there exists a finite set of seed strings, denoted as \( S_0 \), such that any tagging that is compatible with \( S_0 \) is also compatible with \(\lang\).
\end{theorem}
  
\begin{proof}
  The strategy is to first construct a set of seed strings that provide information of the oracle call and return symbols, then extend the set with more seed strings to exclude the incompatible tagging functions.

  Initialize \( S_0 \) as an empty set. For each oracle call-return symbol pair \((\ac,\bc)\), pick a seed string \(s\) that contains a nesting pattern \((x,y)\) where \(\ac\) is unmatched in \(x\), and \(\bc\) is unmatched in \(y\). Incorporate \(s\) into \(S_0\). Note that if there is no such nesting pattern for \((\ac,\bc)\), then it is easy to show that  \((\ac,\bc)\) are ``redundant'' in that they can be treated as plain symbols, which does not change the language with tagging removed.

  Then, for each tagging \( t \) that can be found in \( S_0 \), but is incompatible with a nesting pattern of a certain string \( s \) in \( \lang \), include string \( s \) in \( S_0 \). Given that the set of such tagging functions is finite, \(S_0\) remains a finite set.

  In conclusion, given such \(S_0\), Algorithm~\ref{alg:backtracking_tagging} can at least find the oracle tagging (with redundant tagging removed), by iteratively selecting the oracle call-return pair for each nesting pattern.
\end{proof}

\section{Proofs of Theorems in Section~\ref{sec:learn_token_tag}}
\label{app:proofs_of_learn_program_input}

This section is organized as follows. Lemma~\ref{app:lemma:token_in_pattern} shows that given a nesting pattern \( uxzyv \),
for sufficiently large \( k \), string \( ux^kzy^kv \) contains an unmatched oracle call token in \( x^k \), and contains an unmatched oracle return token in \( y^k \).
Since such \( k \) varies among strings, Lemma~\ref{app:lemma:token_in_pattern2} bounds \( k \) with \( k\leq 2 \) with the help of Exclusivity.
Moving on, Lemma~\ref{app:lemma:oracle_token_vpl} shows that for oracle language \( \lang \) and oracle tokenizer \( \tk_\oracle \), 
\( \tilde{\lang}_{\tk_\oracle} \) is a VPL over \( \tilde{\Sigma}_\oracle \).
Lemma~\ref{app:lemma:cr_covered_by_nesting_pattern} shows that token-based matching rules lead to nesting patterns. Based on Lemma~\ref{app:lemma:cr_covered_by_nesting_pattern}, Lemma~\ref{app:lemma:finite_oracle_unmatched_tokens} bounds the number of unmatched tokens between \(\tk(s)\) and \(\tk_\oracle(s)\). Following Lemma~\ref{app:lemma:oracle_token_vpl} and~\ref{app:lemma:finite_oracle_unmatched_tokens}, Theorem~\ref{app:thm:token_vpl} proves that a compatible tokenizer converts the oracle language into a VPL. We conclude this section with Lemma~\ref{app:lemma:finite_seed_strings_for_tokenization}, which shows that there exists a finite set of seed strings that allows V-Star to find a compatible tokenizer.

\begin{lemma}[Matched Tokens in Nesting Patterns, I]
  \label{app:lemma:token_in_pattern}
  Given oracle language \( \lang \), oracle tokenizer \( \tk \), 
  and a VPG \( G=(\Sigma,V,P,L_0) \) for the oracle language, 
  for each nesting pattern \( s=uxzyv \), 
  there exists an oracle matching rule in \( P \), 
  denoted as \( L \to h_a{A}h_bB \), 
  and exists a positive number \( k \), 
  so that tokens \(h_a\) and \(h_b\) are included in \(x^k\) and \( y^k \), respectively.
\end{lemma}
\begin{proof}
  For each \( k>0 \), consider the tokenization of \( ux^kzy^kv \), denoted as \(l_k=\tk(ux^kzy^kv)\). 

  Let \(l_k^x\) be the token sequence in \(l_k\), i.e., each token in \(l_k^x\) overlaps \(x^k\).
  
  We say a token \emph{covers} a string, if the string is a substring of the string captured by the token.
  For example, if token \(h\) in \(l_k\) captures string \(ux\), then \(h\) covers \(x^1\).
  
  We show that there exists an upper bound, denoted as \(N_x\), such that for each \(k\) and each token \(h\in l_k^x\), if \(h\) covers \(x^i\), then \(i\leq N_x\). Otherwise, notice that there are only a finite number of tokens, therefore, there exists a token \(h\) such that \(h\) starts at a character in \(ux^k\), and can end at a character in either infinite locations in \(x^k\) for \(k>0\), or infinite locations in \(y^k\) for \(k>0\) (could be both). 
  In the first case, one can see that there exist two strings \(s_1\), \(s_2\) of \(h\), such that 
  \begin{align*}
    s_1&=vx^ix_1 \\
    s_2&=vx^jx_1
  \end{align*}
  where \(v\) is a suffix of either \(u\) or \(x\), \(i\) and \(j\) are two numbers such that \(i \neq j\), and \(x_1\) is a prefix of \(x\). Apparantly, \(s_1\) and \(s_2\) can be exchanged, which violates that \((x,y)\) is a nesting pattern. The second case could be similarly invalidated.

  Similarly, we can prove another upper bound \(N_y\) for \(y\). Let \(N\) be \(\max(N_x,N_y)\).

  With upper bound \(N\), we can consider a new nesting pattern \(u'xz'yv'\), where \(u'=ux^N\), \(z'=x^Nzy^N\), and \(v'=y^Nv\):
  \[ (ux^N)x^k(x^Nzy^N)y^k(y^Nv). \]
  By doing this, we exclude the first and last tokens in \(l_k\) that only partially overlap with \(x^k\). From now on, we assume \(l_k\) is contained in \(x^k\) for each \(k\).
  
  For each positive number \(K\),
  define language \( \lang_K \) as \(\{ w \in l_k \mid k>K \} \), and token-based language \( \lang_K' \) as \(\{ l_k \mid k>K \} \). 
  Apparantly, \( \lang_K \) is not a regular language.
  Based on this, we can show that there exists token list \( l \in \lang_K' \) that contains unmatched call or return tokens. Otherwise, each sequence \(l \in \lang_K'\) can only contain plain tokens or well-matched tokens.
  We only need to show that the depth of the nested well-matched tokens is bounded for all \(k\), then the language \( \lang_K \) is a regular language, a contradiction.

  To show that the depth is bounded, notice that otherwise, we would have a nesting pattern in \(x^k\) for certain \(k\). Denote the pattern as \((g,h)\). Nesting pattern \((g,h)\), by definition, can be replaced by \((g^{i},h^{i})\) for any \(i>0\) in \(x^k\). This replacement extends \(x^k\) to \(x^{k'}\) for certain \(k' > k\).
  However, since \((x,y)\) is a nesting pattern, \(ux^{k'}zy^{k}v\) should be invalid, a contradiction.
  
  Therefore, we have shown that for each \(K\), there exists \(k\) and \(l_k\), such that \(k>K\) and \(l_k\) contains a non-plain token. In other words, when \( k \) goes to infinity, \( x^k \) contains an infinite number of unmatched call or return tokens. 
  However, the number of unmatched return tokens is bounded in \( x^k \) for \(k>0\), otherwise, since \( u \) is fixed, not enough call tokens can match those return tokens.
  Therefore, the number of unmatched call tokens is unbounded in \( x^k \) for \( k>0 \). This means for a sufficiently large \(k\), a call token in \(x^k\) must be matched with a return token in \(y^k\). We thus have proven the lemma.
\end{proof}

\begin{lemma}
  [Matched Tokens in Nesting Patterns, II]
  \label{app:lemma:token_in_pattern2}
  Based on Lemma~\ref{app:lemma:token_in_pattern},
  assume Exclusivity.
  For any nesting pattern \( uxzyv \), 
  there exists a matching rule, denoted as \( L \to \ta{A}\tb B \), such that \( s_a\in \ta \) is a substring of \( x^2 \), and \( s_b\in \tb \) is a substring of \( y^2 \).
\end{lemma}
\begin{proof}
  From Lemma~\ref{app:lemma:token_in_pattern}, we know that a token \(h_a\) is in \(x^k\) for certain \(k\).
  We can therefore let \( x= x_1x_2=x_1'x_2' \), so that \(s_a=x_2x^ix_1'=x_2(x_1x_2)^ix_1'\) for certain \(i\). Consider \(x_2\); there are two cases.
  
  If \(x_2\) is empty, then \(s_a=x_1^ix_1' =x^ix_1'=(x_1'x_2')^ix_1'\). 
  Notice that we must have \(i\leq 1\), otherwise either \(x_1'\) or \(x_2'\) becomes both a prefix and an infix of \(s_a\), which violates Exclusivity. Therefore, \(s_a = x_1'\) or \(x_1'x_2'x_1'\). Apparantly, \(s_1\) is a substring of \(x^2\).

  If \(x_2\) is nonempty, \(i\) must be zero, otherwise \(x_2\) is both a prefix and an infix, a violation of Exclusivity. In this case, we have \( s_a=x_2x_1' \), also a substring of \(x^2\).

  In conclusion, we know that \( s_a \) is a substring of \( x^2 \).
  The reasoning is quite similar for the return token, and we omit it for brevity.
\end{proof}

\begin{lemma}\label{app:lemma:oracle_token_vpl}
  For oracle language \( \lang \) and oracle tokenizer \( \tk_\oracle \), 
  \( \tilde{\lang}_{\tk_\oracle} \) is a VPL over \( \tilde{\Sigma}_\oracle \).
\end{lemma}
\begin{proof}
  We prove by building a VPA for language \( \tilde{\lang}_{\tk_\oracle} \).

  First, since the language \( \{ \tk_\oracle(s) \mid s\in\lang \} \) is a VPL over \( T_{\tk_\oracle} \), we have a VPA for it, denoted as \( \vpa_\oracle \).

  Then, since each token \( \tk \) is a regular language, denote its finite state automaton as \( \fsa_t \).

  Now we build a VPA \( \tilde{\vpa} \) for language \( \tilde{\lang}_{\tk_\oracle} \) by replacing each token in \( \vpa_\oracle \) with its FSA \( \fsa_t \).

  First, the set of states is the union of states in \( \vpa_\oracle \) and \( \fsa_t \) for \( \tk\in T_{\tk_\oracle}\). 
  
  Then, the transitions are defined as follows.
  \begin{enumerate}
    \item 
    We retain transitions \( p \xrightarrow{i} p' \) in each FSA \( \fsa_{h_c} \) for plain token \( h_c \).
    \item 
    For transition \( q\xrightarrow{h_c} q' \), we add transition \( (q,\epsilon)\to S_{h_c,0} \) for the start state \( S_{h_c,0} \) and transitions \( (E_{h_c},\epsilon)\to q' \) for each acceptance state \( E_{h_c} \) in \( \fsa_{h_c} \).
    \item
    For transition \( q\xrightarrow{h_a, \text{ push }} q' \), we add transition \( q \xrightarrow{\ac, \text{ push } (q,\ac)} S_{h_a,0} \), where \( \ac \) is the call token corresponding to \( h_a \), and we add transitions \( (E_{h_a},\epsilon)\to q' \) for each acceptance state \( E_{h_a} \) in \( \fsa_{h_a} \).
    \item
    For transition \( q\xrightarrow{h_b, \text{ pop } (q'',h_a')} q' \), we add transition \( q \xrightarrow{\bc, \text{ pop } (q'',\ac')} q' \), where \( \ac' \) is the call token corresponding to \( h_a' \), and we add transitions \( (E_{h_b},\epsilon)\to q' \) for each acceptance state \( E_{h_b} \) in \( \fsa_{h_b} \).
  \end{enumerate}
  
  Apparantly, a string \( s\in\lang \) must be accepted by \( \tilde{\vpa} \).
  On the otherhand, if a string \( s \) is accepted by \( \tilde{\vpa} \), then it leads to a valid token sequence \( l \) and \( s\in l \), therefore, 
  \( s \) is also a valid string \( \lang \).

  In conclusion, \( \tilde{\vpa} \) is a VPA that accepts \( \tilde{\lang}_{\tk_\oracle} \), which means language \( \tilde{\lang}_{\tk_\oracle} \) is a VPL.
\end{proof}

\begin{lemma}[Call and Return Tokens in Nesting Patterns]
  \label{app:lemma:cr_covered_by_nesting_pattern}
  For any oracle VPL \( \lang \), for each string \( s \in \lang \), if \( s \) is derived by repeatedly applying a matching rule which exposes a recursion, i.e.,
  \[ L_0 \to^* pLq \to p(\ta A\tb B)q \to^* p(s_a uLv s_bB)q \to^* ps_a us_Lv s_bs_Bq =s \]
  where \( s_a \in \ta  \) and \( s_b \in \tb  \), and \( s_L,s_B\in\lang \) are strings derived from nonterminals \( L \) and \( B \), respectively, and \(u,v\in\lang\). Then there exists a nesting pattern \( (x, y) \) for \( s \) where \( s_a \) is a prefix of \( x \) and \( s_b \) is a substring of \( y \).
\end{lemma}
\begin{proof}
  Consider the iterative application of the derivation \( pLq\to^*p(s_a uLv s_bB)q \) to \( L \). This leads to the deduction \( pLq\to^*p(s_a u)^kL(v s_bs_B)^kq \). 
  To show that the pair \( (s_a u, v s_bB) \) represents a nesting pattern, we only need to prove that \( ux^kzy^jv \) is invalid when \(k\neq j\) (both \(\geq 0\)).

  Let the oracle tokenizer be \( \tk \).
  For \(k>0\), we tokenize string \( s_3 = p(s_a u)^ks_L(v s_bB)^kq \) as:
  \begin{align*}
    \tk(s_3) &= \tk ( p(s_a u)^k s_L (v s_b s_B)^k q ) \\
    &= \tk(p) \ta  ( \tk(u) \ta  )^{k-1} \tk(u s_L v) \tb  ( \tk(s_B v) \tb  )^{k-1} \tk(s_B q) \quad\text{(By Tokenization Consistency)} && 
  \end{align*}
  Notice that \(u\), \(s_L\) and \(v\) are independent strings of token list. Therefore, we have 
  \[ \tk(u s_L v) = \tk(u)\tk(s_L)\tk(v). \]

  We next tokenize string \( s_4 = p(s_a u)^ks_L(v s_bs_B)^jq \), for \( k\neq j \) (both \(>0\)):
  \begin{align*}
    \tk(s_4) &= \tk ( p(s_a u)^k s_L (v s_b s_B)^j q ) \\
              &= \begin{cases} 
                  \tk ( p(s_a u)^k s_L v s_b( s_B v s_b)^{j-1} s_B q ) & \text{if } k > j \\
                  \tk ( p(s_a u)^k s_L v s_b( s_B v s_b)^{k-1} s_B (v s_bs_B)^{j-k} q ) & \text{if } k < j 
                 \end{cases} \\
              &= \begin{cases} 
                  \tk ( p(s_a u)^k s_L v s_b( s_B v s_b)^{j-1} ) \tk ( s_B q ) & \text{if } k > j \\
                  \tk ( p(s_a u)^k s_L v s_b( s_B v s_b)^{k-1}) \tk (  s_B (v s_bs_B)^{j-k} q ) & \text{if } k < j 
                 \end{cases} \\
              &= \tk(p) \ta  ( \tk(u) \ta  )^{k-1} \tk(u s_L v) \tb  ( \tk(s_B v) \tb  )^{j-1} \tk(s_B q)
  \end{align*}
  We applied Tokenization Consistency in the last step above. Apparantly,
  \(\tk(s_4)\) is invalid because of imbalanced call and return tokens. 
  There are two cases left, where either \(k\) or \(j\) equals \(0\).

  When \(k=0\) and \(j>0\), \(s_4= ps_L(vs_bs_B)^jq \). Assume \(s_4\) is valid; we can tokenize \(s_4\) as
  $$
  \begin{aligned}
    \tk(s_4) &= \tk(ps_L(vs_bs_B)^jq) \\
    &= \tk(ps_L v s_b( s_B v s_b)^{j-1} s_B q) \\
    &=\tk(p)\tk(s_L)\tk(v) \tb (\tk( s_B v) \tb )^{j-1}\tk(s_B)\tk(q). 
  \end{aligned}
  $$
  Again, \(\tk(s_4)\) is apparantly invalid. The case of \(k>0\) and \(j=0\) is similar and we omit it for brevity.

  Therefore, we have shown that the pair \((s_a u, v s_bB)\) is a nesting pattern.
\end{proof}

\begin{lemma}\label{app:lemma:finite_oracle_unmatched_tokens}
  Given a compatible tokenizer \( \tk \), there exists an upper bound \( N_\tk \),
  so that for each string \( s\in\tilde{\Sigma}^* \), if \( \tilde{s} \) is well-matched and there exists context strings \( (w,w') \) such that \( wsw'\in\lang \), then string \( s \)  contains at most \( N_\tk \) unmatched oracle call or return tokens.
\end{lemma}
\begin{proof}
  The proof parallels the proof of Lemma~\ref{app:lemma:finite_oracle_unmatched_symbols}; we show that this would otherwise result in an infinite number of equivalence classes for \( \sim_{\tk_\oracle} \), contradicting to that \( \tilde{L}_{\tk_\oracle} \) is a VPL, proven by Lemma~\ref{app:lemma:oracle_token_vpl}.

  If for any number \( K \), there exists a \( \tk \)-well-matched string \( s \), such that \( s \) can contain at least \( K \) unmatched call token, then, by Lemma~\ref{app:lemma:cr_covered_by_nesting_pattern}, a set of nesting patterns \( (x_i,y_i) \) appear in \( wsw' \), and therefore in \( \tilde{w}\tilde{s}\tilde{w}' \).
  Then, since \( \tk \) is compatible with \( \lang \), a call-return token pair \( \callsym{c},\retsym{d} \) appear in each \( (\tilde{x_i},\tilde{y_i}) \).
  The rest of the analysis is similar to that of Lemma~\ref{app:lemma:finite_oracle_unmatched_symbols}, and we omit it for brevity.
\end{proof}

\begin{theorem}\label{app:thm:token_vpl}
  Given oracle language \( \lang \) and oracle tokenizer \( \tk_\oracle \),
  for each tokenizer \( \tk \) that is compatible with the oracle language \( \lang \), language \( \tilde{\lang}_{\tk} \) is a VPL.
\end{theorem}
\begin{proof}
  By Lemma~\ref{app:lemma:finite_oracle_unmatched_tokens},
  for certain \( k\leq N_\tk \),
  we can represent \( \tilde{q}_\oracle \) as \( q_1\ac_1 q_2\ac_2\dots \ac_{k-1} q_k \), where \( q_{i,i\in[1..k]} \) is well-matched under \( \tk_\oracle \). Now, we replace each  \( q_{i,i\in[1..k]} \) with its representative in the equivalence class of \( \sim_{\tk_\oracle} \), and get \( q' \) that \( \tilde{q'}\sim_{\tk} \tilde{q} \). Since \( \sim_{\tk_\oracle} \) has a finite number of equivalence classes, let the length of the longest representative be \( l \), the length of \( q' \) is bounded by \( l\times k+k+1 \).

  Therefore, the congruence relation \( \sim_\tk \) has a finite number of equivalence classes, which shows \( \tilde{\lang}_{\tk} \) is a VPL.
\end{proof}

\begin{lemma}[Finite and Sufficient Seed Strings]
  \label{app:lemma:finite_seed_strings_for_tokenization}
  There is a finite set of seed strings \( S \), with which we can find a tokenizer that is compatible with the oracle language \( \lang \) using Algorithm~\ref{alg:backtracking_token}.
\end{lemma}
\begin{proof}
  The construction of \(S\) involves two phases: we first identify strings that reveal the oracle call and return tokens, then augment this set to exclude taggings incompatible with \(\lang\).

  Starting with an empty set \(S_0\), for each oracle call-return token pair \((\ta,\tb)\), we select a seed string \(s\) that contains a nesting pattern \((x,y)\) where \(\ta\) and \(\tb\) are respectively unmatched in \(x\) and \(y\), then we include \(s\) in \(S_0\). We denote the new set as \(S_1\).
  Note that, similar to Lemma~\ref{app:thm:construct_seed_strings}, if there is no such nesting pattern for \((\ta,\tb)\), then it is easy to show that  \((\ta,\tb)\) are ``redundant'', in that they can be treated as plain tokens, which does not change the language with tagging removed.

  Subsequently, we extend \(S_1\) with strings \(s\) from \(\lang\) that is incompatible with certain tokenizer that can be found by Algorithm~\ref{alg:backtracking_token} based on \(S_1\). We denote the new set as \(S_2\). \(S_2\) is still finite, since Algorithm~\ref{alg:backtracking_token} can only find a finite number of tokenizers from \(S_1\).

  For each string \(s\) in \(S_2\), we modify \(s\) by replacing call and return tokens \(\ta\) and \(\tb\) with strings \(s_a\) and \(s_b\) from Token Fixed Prefix and Suffix, respectively. The new set is our targeted \(S\).

  Given \(S\), we now show that the oracle tokenizer \(\tk\) is a possible return of Algorithm~\ref{alg:backtracking_token}. Intuitively, the oracle token pair can be incrementally added to the hypothesis partial tokenizer \(D\), and the tokenization \(\tokenit(D,s)\) for any valid string \(s\) maintains well-matched.
  
  By Lemma~\ref{app:lemma:token_in_pattern2}, each oracle token pair is contained in \((x^2,y^2)\) for certain nesting pattern \((x,y)\). Therefore, from the first nesting pattern \((x,y)\), 
  Algorithm~\ref{alg:backtracking_token} could find \(((q,g),(q',g'))\) where \(((q,g),(q',g'))\) belongs to an oracle token pair. Then, because of the Token Fixed Prefix and Suffix assumption, the lexical rules for the two paired tokens are learned accurately.
  Denote this tokenizer that contains only one token-pair as \(D_1\).
  Consider Algorithm~\ref{alg:tokenize}.
  At line 4, we construct a new match \( m \). 
  Now, we show that \(m\) must be the match given by the oracle tokenizer.
  Firstly, \(m\) must correspond to an oracle match, otherwise it will be filtered out by \(k\)-Repetition. Then, by Unique Pairing and Separation, \(m\) must correspond to the oracle call/return token. Therefore, \(\tokenit(D_1,s)\) contains only well-matched oracle tokens, thus is well-matched.

  A similar analysis can be done to show that as long as \(D\) contains only oracle token pairs, \(\tokenit(D,s)\) is well-matched for any valid string \(s\). We thus have proved that Algorithm~\ref{alg:backtracking_token} can at least find the oracle tokenizer given \(S_0\) (with redundant oracle tokens removed), by iteratively selecting the oracle token pair for each nesting pattern.

  In conclusion, there exists a finite set of seed strings, where Algorithm~\ref{alg:backtracking_token} can find a compatible tokenizer.
\end{proof}

\newpage

\bibliographystyle{ACM-Reference-Format}

\end{document}